\documentclass[11pt]{article}
\usepackage[top=2.5cm,right=2.5cm,bottom=2.5cm,left=2.5cm]{geometry}
\usepackage{graphicx,subfig,float,color}
\usepackage{amssymb,amsmath}
\allowdisplaybreaks 

\usepackage[pdftex,pdfauthor={Pavithran S Iyer}, pdftitle={\bf Complexity Error Correction}]{hyperref}
\hypersetup{
    colorlinks,
    citecolor=blue,
    filecolor=blue,
    linkcolor=blue,
    urlcolor=blue
}

\everymath{\displaystyle}
\numberwithin{equation}{section}

\newcommand{\slfrac}[2]{\left.#1\middle/#2\right.}

\numberwithin{table}{section}

\title{Hardness of decoding quantum stabilizer codes}
\author{Pavithran Iyer and David Poulin\footnote{\href{mailto:pavithran.iyer.sridharan@usherbrooke.ca}{pavithran.iyer.sridharan@usherbrooke.ca} , \href{mailto:David.Poulin@usherbrooke.ca}{David.Poulin@usherbrooke.ca} \newline D\'epartement de Physique, Universit\'{e} de Sherbrooke, Qu\'{e}bec, Canada - J1K2R1}}

\numberwithin{equation}{section}

\DeclareMathOperator{\wt}{\textsf{wt}}
\DeclareMathOperator{\wtt}{\textsf{wt}_2}
\DeclareMathOperator{\poly}{\textsf{poly}}

\DeclareMathOperator{\WE}{\textsf{WE}}

\DeclareMathOperator{\QMLD}{\textsf{QMLD}}
\DeclareMathOperator{\DQMLD}{\textsf{DQMLD}}
\DeclareMathOperator{\Prob}{\textsf{Prob}}
\DeclareMathOperator{\argmax}{\textsf{ArgMax}}

\DeclareMathOperator{\PH}{\textsf{PH}}
\DeclareMathOperator{\sharpP}{\textsf{\#P}}

\renewcommand{\P}{\textsf{P}}
\DeclareMathOperator{\NP}{\textsf{NP}}
\DeclareMathOperator{\NPHard}{\textsf{NP-Hard}}
\DeclareMathOperator{\NPComplete}{\textsf{NP-Complete}}

\DeclareMathOperator{\sharpPHard}{\textsf{\#P-Hard}}
\DeclareMathOperator{\sharpPComplete}{\textsf{\#P-Complete}}

\DeclareMathOperator{\FP}{\textsf{FP}}

\def\ket#1{| #1 \rangle}

\def\cC{\mathcal{C}}

\def\cG{\mathcal{G}}
\def\cH{\mathcal{H}}

\def\cJ{\mathcal{J}}

\def\cL{\mathcal{L}}

\def\cN{\mathcal{N}}
\def\cO{\mathcal{O}}

\def\cQ{\mathcal{Q}}

\def\cS{\mathcal{S}}
\def\cT{\mathcal{T}}

\def\cZ{\mathcal{Z}}

\def\sfD{\textsf{D}}

\def\sfM{\textsf{M}}

\def\sfT{\textsf{T}}

\def\sfV{\textsf{V}}

\newtheorem{definition}{Definition}[section]
\newtheorem{theorem}{Theorem}[section]
\newtheorem{lemma}{Lemma}[section]

\newenvironment{proof}[1][Proof\thinspace:]{\begin{trivlist}
\item[\hskip \labelsep {\bfseries #1}]}{\hfill \large{$\Box$}\end{trivlist}}

\newenvironment{proofMainThm}[1][\textbf{Proof of Thm.~({}\ref {main-theorem}{})}\thinspace:]{\begin{trivlist}
\item[\hskip \labelsep {\bfseries #1}]}{\hfill \large{$\Box$}\end{trivlist}}

\begin{document}
\maketitle
\begin{abstract}
In this article we address the computational hardness of optimally decoding a quantum stabilizer code. Much like classical linear codes, errors are detected by measuring certain check operators which yield an error syndrome, and the decoding problem consists of determining the most likely recovery given the syndrome.  The corresponding classical problem is known to be $\NPComplete$ \cite{mc-ber-78,V97a}, and a similar decoding problem for quantum codes is also known to be $\NPComplete$ \cite{heish-gall-11,kuo-lu-12}. However, this decoding strategy is not optimal in the quantum setting as it does not take into account error degeneracy, which causes distinct errors to have the same effect on the code. Here, we show that optimal decoding of stabilizer codes is computationally much harder than optimal decoding of classical linear codes, it is $\sharpPComplete$.

\vspace{0.3cm}

\end{abstract}

\section{Introduction}

In his seminal papers that gave birth to the field of information theory, Shannon showed that the capacity of a channel could be achieved using codes whose codewords are random bit strings \cite{Sha48a}. Despite this optimality, random codes have no practical use because we do not know how to decode them efficiently, i.e. in a time polynomial in the number of encoded bits. In 1978, Berlekamp, McEliece, and van Tilborg \cite{mc-ber-78} (see also \cite{V97a}) showed that decoding a classical linear code is an $\NPComplete$ problem, which strongly indicates that no efficient algorithm will ever be found to decode generic classical codes. A central problem in coding theory therefore consists in designing codes that retain the essential features of random codes, but yet have enough structure to be efficiently (approximately) decoded.

Quantum information science poses additional challenges to coding theory. While the stabilizer formalism establishes many key parallels between classical and quantum coding \cite{got-phd-97,CRSS97a}, important distinctions remain. On the one hand, quantum code design is obstructed by the additional burden that check operators must mutually commute. For that reason, it has proven difficult to quantize some of the best families of classical codes, such as low density parity check (LDPC) codes \cite{FM01a,MMM04a,COT07a,TZ09a} and turbo codes \cite{poulin-tillich-09}. On the other hand, quantum codes can be {\em degenerate}, which means that distinct errors can have the same effect on all codewords. This changes the nature of the decoding problem \cite{WHP03a,poulin-06,dav-gui-10,dav-emi-12}, and our goal here is to explore how degeneracy impacts the computational complexity of the decoding problem.

The conceptually simplest method to decode a stabilizer code is to ignore error degeneracy and to proceed as with classical linear codes: amongst all errors consistent with the error syndrome, find the one with the highest probability. We call this decoding method {\em Quantum Maximum Likelihood Decoding} ($\QMLD$). It was shown in \cite{heish-gall-11,yueh-13,kuo-lu-12} that $\QMLD$ is $\NPComplete$.

In the presence of degeneracy however, errors naturally fall into equivalence classes, with all errors in the same class having the same effect on all codewords.  The optimal decoding method searches over all equivalence classes of errors that are consistent with the error syndrome, the one with the largest probability. The probability of a class of error is simply the sum of the probabilities of all errors it contains. We call this decoding method {\em Degenerate Quantum Maximum Likelihood Decoding} ($\DQMLD$). Our main result is the following theorem.

\medskip
\noindent{\bf Thm.~({}\ref {main-theorem}{})} (Informally) 
$\DQMLD \in \sharpPComplete$ {\it up to polynomial-time Turing reduction.}
\medskip

We need Turing reduction since decoding is not a counting problem, while problems in $\sharpP$ consist in counting the number of solutions to a decision problem in $\NP$. Our result can be understood intuitively from the fact  that in order to compute the probability associated to an equivalence class, $\DQMLD$ must determine how many errors of a given weight belong to an equivalence class, hence the need to count. Our proof uses a reduction from the problem of evaluating the \emph{weight enumerator} of a classical (binary) linear code, which was shown by Vyalyi to be $\sharpPComplete$ \cite{vyali-03,B11d}. 

The rest of this paper is organized as follows. For self-containment, the next two sections provide elementary introductions to computational complexity and to the stabilizer formalism. Sec.~({}\ref {decoding-problem}{}) formally defines the decoding problem with a particular emphasis on the role of degeneracy. This section also contains an informal discussion on the importance of error degeneracy and how it impacts the decoding problem in specific settings. Sec.~({}\ref {sec-hardness}{}) presents the main result, which is proved in Sec.~({}\ref {sec-red-over}{}); the expert reader can jump directly to these two sections. The conclusion proposes possible extensions of the present work.

\section{Computational Complexity} \label{complexity-intro}
Two key resources required to solve any problem in computer science are space and time. One way of classifying problems is based on the runtime of a corresponding algorithm for that problem. This time is expected to depend on the size of the input, $n$, to the problem. Instead of precisely denoting the runtime as a function $f(n)$, which would depend on the underlying hardware of the solver, it is largely classified by its \emph{limiting behavior} as: $\mathcal{O}(\log_{2}n), \mathcal{O}(n^{k}), \mathcal{O}(2^{n})$ and so on. Consequently, the class of problems for which there exists an algorithm which runs in time $\mathcal{O}(n^{k})$, on an input of size $n$, $k$ being a constant independent of $n$, is called $\P$. Note that any problem in $\P$ is a decision problem, i.e, one for which the solution space is binary. If the problem is formulated to produce an output string, then the existence of a polynomial time algorithm classifies this problem in the class $\FP$. There are problems to which any witness or certificate for a solution can be verified in polynomial time. These fall into the class called $\NP$ and clearly, $\P \in \NP$.

A problem $P_{1}$ is \emph{at least as hard as} $P_{2}$ if we could use the solver for $P_{1}$ to solve $P_{2}$. This is formalized by the notion of a \emph{reduction}, which enables the classification of problems that are harder than a particular class, thereby introducing $\NPHard$ and $\NPComplete$, the latter being a class of hardest problems in $\NP$.

Besides decision problems, a category of problems involve enumerating the elements of a set. When this set is in $\NP$, the corresponding enumeration problem is classified as $\sharpP$ \cite{V79a,AB09}.
\begin{definition}[Counting class]{$\sharpP$:} \label{sharpP-defn}
A counting function $f:\{0,1\}^{\ast}\to\mathbb{N}$ is in $\sharpP$ if $\exists$ a problem $P_{1}\in\NP$ such that $f(x) = \left|\{y\in\{0,1\}^{\poly(|x|)} \thickspace : \thickspace P_{1}(x,y) = 1\}\right|$, $\forall\thinspace x\in\{0,1\}^{\ast}$.
\end{definition}
Hence, a function in $\sharpP$ counts the number of solutions to a corresponding decision problem in $\NP$. One can also compare the relative hardness of two counting problems as is done for the $\NP$ case, using the notion of counting reductions.
\begin{definition}[Counting reduction]
Let $f,g$ be counting functions. The function $f$ is Turing reducible to $g$, if $\forall\thinspace x \in\{0,1\}^{\ast}$, $f(x) \in \FP^{\thinspace g}$.
\end{definition}
That is, $f(x)$ can be computed in polynomial-time with polynomially many queries to an oracle for $g(x)$, for any $x$. Subsequently, analogous to $\NPComplete$, a notion of reduction for counting functions defines the hardest problems in $\sharpP$, as the class $\sharpPComplete$ \cite{AB09}.
\begin{definition}[Counting complete]{$\sharpPComplete$}:
A counting function $f$ is in $\sharpPComplete$ iff $f$ is in $\sharpP$ and $g \thinspace \in \thinspace \FP^{\thinspace f}$, $\forall\thinspace g\thinspace\in\thinspace \sharpP$.
\end{definition}
The last criterion can also be replaced with: $g\thinspace \in\thinspace \FP^{f}$, for some $g\thinspace \in\thinspace \sharpPComplete$. We will prove $\sharpPComplete$-ness of the problem of our concern, by adhering to such a recursive definition.

The level of hardness of a $\sharpPComplete$ problem can be appreciated by a consequence of a result shown by Toda \cite{T91a,AB09}, stating that a polynomial time algorithm that is allowed a {single access} to a $\sharpP$ oracle, can solve any problem in $\PH$, i.e, $\P^{\sharpP} = \NP \cup \NP^{\NP} \cup \NP^{\NP^{\NP}} \dots $.

A particular example of a counting function that will be of interest in the coming sections is the \emph{weight enumerator} function for a linear code.
\begin{definition}[Weight enumerator problem] {$\WE$:} \label{WE-def-code}
Given a $(n,k)$ linear code $\cC$ and a positive integer $i$, compute $\WE_{i}(\cC) = |\{c\in\cC : |c| = i\}|$, the number of codewords of Hamming-weight $|c| = i$.
\end{definition}
The corresponding decision problem, which is to determine if there exists a codeword in $\mathcal{C}$ of weight $i$, is known to be in $\NPComplete$ \cite{mc-ber-78,barg97}. This immediately implies that $\WE_{i}(\cC) \in \sharpP$, and furthermore, it is known to be complete for this class \cite{vyali-03,B11d}.

\begin{theorem}{Hardness of $\WE$:} \label{WE-sharpP-complete}
For a linear code $\mathcal{C}$ and $i= \Omega(\textsf{polylog}(n))$, $\WE_{i}(\mathcal{C}) \in \sharpPComplete$.
\end{theorem}
Hence to show that a problem of computing $f$ is in $\sharpPComplete$, in addition to its membership in $\sharpP$, it suffices to show that for any $(n,k)$ linear code $\cC$, $\WE_{i}(\mathcal{C})$, $i\in o(\textsf{polylog}(n))$ can be computed in polynomial-time by allowing at most polynomially many queries to an oracle for computing $f$.

\section{Stabilizer codes} \label{sec-stab-codes}

In this section, we will provide the necessary background material on stabilizer codes, see \cite{got-phd-97,NC00} for more complete introductions. In the setting of quantum coding, information is encoded into $n-$qubit states in $\cH^{n}_{2} = (\mathbb C^2)^{\otimes n}$.  Errors are operators acting on this space, modeled as elements of the  \emph{Pauli group} ${\cG}_{n}$. 
\begin{definition}[Pauli group]  \label{def-pauli-group}
The Pauli group on $1$-qubit, denoted by $\mathcal{G}_{1}$ is a matrix group over $\mathbb{C}^{2}$, generated by the Pauli matrices; $\mathcal{G}_{1} = \{\pm X, \pm Y, \pm Z, \pm i X, \pm i Y, \pm i Z, \pm \mathbb{I}\} = \langle i, X, Y, Z\rangle$. The Pauli group on $n$-qubits is the group generated by the set of Pauli matrices acting on $n$-qubits, denoted by $\mathcal{G}_{n} = \mathcal{G}^{\otimes n}_{1} = \{\epsilon P_{1} \otimes \cdots \otimes P_{n} | \epsilon\in\{\pm 1, \pm i\} \thickspace, \thickspace\thickspace P_{i}\in \{\mathbb{I},X, Y, Z\}\}$.
\end{definition}
Though the definition of the Pauli group in Def.~({}\ref {def-pauli-group}{}) contains the scalars $\{\pm 1, \pm i\}$, they are often considered unimportant for error detection or correction as they do not affect the error syndrome nor the error correction. This enables us to define the effective Pauli group by identifying operators that are related by a multiplicative constant in $\{\pm 1, \pm i\}$, denoted by $\overline{\cG}_{n} = \cG_{n}/{\{\pm 1, \pm i\}}$. The number of qubits affected by an error $E\in\overline{\cG}_{n}$ is the number of non-identity components in the tensor product form of $E$, and is called the {\em weight} of the error, denoted by $\wt(E)$.

A quantum stabilizer code is defined by a set of check operators, which are also elements of the Pauli group.
\begin{definition}[Stabilizer codes] \label{def-stabilizer-codes}
A stabilizer code is a subspace of the $n$-qubit vector space $\cH_2^n$, described as the common +1 eigenspace of an Abelian and Hermitian subgroup of ${\cG}_{n}$, called the stabilizer subgroup. Hence, for a stabilizer code $\cQ$, the stabilizers are defined by:
\begin{gather}
\mathcal{S} = \{S \in \overline{\cG}^{\thinspace n} \thickspace : \thickspace S|\psi\rangle = |\psi\rangle \thickspace, \thickspace \forall \thinspace |\psi\rangle \in \cQ\}. \label{eq-stabilize}
\end{gather}
Equivalently, we can define the code $\cQ$ in terms of its stabilizers
\begin{equation}
\cQ = \{ \ket\psi \in\cH_{2}^{n}\thickspace :\thickspace S\ket\psi = \ket\psi \thickspace,\thickspace \forall S\in \cS\}
\label{code}.
\end{equation}
\end{definition}
Stabilizer codes are by far the most widely studied codes in the quantum setting. 

\subsection{Operator basis}

When $\cS$ is generated by $n-k$ independent \emph{stabilizer generators} $\{S_j\}_{j=1\ldots n-k}$, the subspace $\cQ$ has dimension $2^k$, so it encodes $k$ logical qubits, and we denote its parameters by $[[n,k]]$. All the elements in $\overline{\cG}_{n}$ that leave the code $\cQ$ globally invariant belong to the \emph{normalizer} of $\cS$ in $\overline{\cG}$, denoted by $\cN(\cS)$. The normalizer $\cN(\cS)$ forms a group generated by $n+k$ generators and $\cS \subset \cN(\cS)$.
However, not all operators in $\cN(\cS)$ necessarily fix every state in $\cQ$. Since the action of $\cS$ is trivial on all code states, we define a subset of operators in $\cN(\cS)$ called the \emph{logical operators} that represent the quotient space $\cL = \cN(\cS)/\cS$, each of which have a distinct action on individual states in $\cQ$. For a shorthand notation, we will identify these representative with elements of $\cL$. This group has $2\thinspace k$ \emph{canonical logical generators} $\{\overline X_j,\overline Z_j\}_{j=1\ldots k}$, such that all generators mutually commute except the pairs $\overline X_j$ and $\overline Z_j$ that anti-commute. The smallest weight of a nontrivial logical operation is the distance of the code, 
\begin{gather}
d =  \min_{E\in\cN(\cS)\backslash\cS}\wt(E).
\end{gather}

Note that the operators defined so far, $\{S_j\}_{j=1\ldots n-k}$ and $\{\overline X_j,\overline Z_j\}_{j=1\ldots k}$ do not generate $\overline{\cG}_{n}$. To complete the basis, we need to define the group of {\em pure errors} $\cT = \cN(\cL)/\cS$, which is also Abelian. We can always find a set of \emph{canonical pure error generators } $\{T_j\}_{j=1\ldots n-k}$ such that $T_j$ commutes with all other pure error generators, all logical generators, and all stabilizer generators except $S_j$ with which it anti-commutes. To summarize, we have the canonical basis of the Pauli group $\{S_i, T_i, \overline X_j, \overline Z_j\}_{i=1\ldots n-k, j=1\ldots k}$ with all commutation relations trivial, except 
\begin{equation}
T_iS_i = - S_iT_i \quad {\rm and}\quad \overline X_j\overline Z_j = -\overline Z_j\overline X_j .
\label{commutations}
\end{equation}
Any Pauli operator  $E \in \bar{\cG}_{n}$  can be expressed as a product of elements in these respective groups :
\begin{gather}
E = T\cdot L\cdot S \label{decomp-pauli-op} \qquad (\text{where $T\in\mathcal{T}$, $L\in\mathcal{L}$ and $S\in\mathcal{S}$}).
\end{gather}
Decomposition in this basis will be particularly useful to formulate the decoding problem. 

\subsection{Degenerate and non-degenerate errors} \label{sec-deg-err}
Two errors $E$ and $E'$ are called \emph{degenerate} if they have an identical effect on all code states, i.e. $E\ket\psi = E'\ket\psi\thinspace, \thinspace \forall \ket\psi\in\cQ$. Given the decomposition in Eq.~({}\ref {decomp-pauli-op}{}) and the definition of the code in Eq.~({}\ref {eq-stabilize}{}), we see that this is only possible if the two errors are related by an element of $\cS$, i.e.,  $E' = E\cdot S$ for some $S\in \cS$. This naturally leads to an equivalence relation between errors, with two errors belonging to the same equivalence class if they are related by an element of $\cS$. In other words, the set of equivalence classes is the quotient space $\overline{\cG}_{n}/\cS \sim \cL\cdot \cT$. We can thus label the equivalence classes by $L,T$ with $L\in \cL$ and $T\in \cT$. For a fixed $T$, the different classes labelled by $L\in\cL$ are referred to as \emph{logical classes} and each class is of size $2^{n-k}$. The class labelled by $L = \mathbb{I}$ and $T = \mathbb{I}$ is $\cS$ itself. Any other logical class can be expressed as a coset of $\cS$, but however they are not groups by themselves.

Note that  degeneracy is unique to the quantum error correction setting. When a bit flip pattern $e$ is applied to a classical bit string $x$, the resulting string $y = x+e$ always differs from the original one, except for the trivial error $e=0^n$. Consequently, each logical class contains only one element. Our main result in this paper shows that accounting for degenerate errors in decoding stabilizer codes greatly increases its computational complexity. 

\subsection{Symplectic representation} \label{sec-symplectic}
There is a one-to-one correspondence between \emph{classical symplectic linear codes} in $\mathbb{Z}^{2n}_{2}$ and stabilizer codes \cite{CRSS97a,got-phd-97}, which follows from a mapping $\eta$ of pauli operators in $\overline{\cG}_{n}$ into binary strings of length $2n$. The mapping is performed by first expressing every $\sfM\in\overline{\cG}_{n}$ in the form:
\begin{gather}
\textsf{M} = \left(\textsf{A}_{1}\otimes \dots \otimes \textsf{A}_{n}\right)\cdot\left(\textsf{B}_{1}\otimes\dots\otimes \textsf{B}_{n}\right) \thickspace, \thickspace\thickspace \textsf{A}_{i}\in\{Z,I\} \thickspace, \thickspace\thickspace \textsf{B}_{j}\in\{X,I\}\thickspace, \thickspace \thickspace 1\leq i,j\leq n \thickspace \label{general-stabilizer-element-X-Z}
\end{gather}
and substituting $X$ and $Z$ by $'1'$ and $\mathbb{I}$ by $'0'$. For instance, $\eta(X\otimes Z\otimes I\otimes Y)=0 1 0 1 1 0 0 1$. This maps $\overline{\cG}^n$ into its \emph{symplectic representation} in $\mathbb{Z}^{2n}_{2}$. Moreover, any two mutually commuting operators  $P,Q \in \cG_n$ are mapped into binary strings $x,y \in \mathbb{Z}^{2n}_{2}$ that are orthogonal under the \emph{symplectic product}, defined by $(x,y) = x \Lambda y^{\text{T}}$ with $\Lambda = I_n\otimes X$. This immediately implies that $\eta(\cN(\cS))$ is a vector space (code) which is the kernel of $\eta(\cS)$ under the symplectic product. The parity check matrix for this classical code is $\eta(\{S_{i}\}_{i=1}^{n-k})$ and its generator matrix is $\eta(\{\overline{X}_{j},\overline{Z}_{j}\}_{j=1}^{k},\{S_{i}\}_{i=1}^{n-k})$. For a $2n$-bit string $b = (z|x)$ expressed as the concatenation of two $n$-bit stings, we denote the inverse mapping $\eta^{-1}(b) = Z^z \cdot X^x$ where we use the shorthand notation $Q^x = Q^{x_1}\otimes Q^{x_2}\otimes\cdots\otimes Q^{x_n}$. 

The mapping $\eta$ implicitly indicates that we consider a $Y-$type Pauli operation as a $Z-$type operation followed by a $X-$ type operation, in other words, as an error of weight two. Consequently, we define the {\em symplectic-weight} of the Pauli error as $\wt_{2}(E) = |\eta(E)|$ with $|\cdot|$ denoting the usual Hamming weight. Clearly, $2\wt(E)\geq \wt_{2}(E) \geq \wt(E)$.

This correspondence with classical linear codes is the key to most of the complexity results \cite{heish-gall-11,yueh-13,kuo-lu-12,vyali-03} involving stabilizer codes, including ours. In many cases it can be used to build parallels to  already known results for classical (linear) symplectic codes.

\subsection{Error model}
An error model assigns probabilities to various errors, which are then used by the decoder to statistically infer what recovery is most likely given the error syndrome. We restrict ourselves to errors from $\overline{\cG}_{n}$, so the corresponding error model is referred to as a \emph{Pauli channel}. There are multiple type of Pauli channels which are often used in studying error correcting codes \cite{preskill-ecc-98,NC00}. Here, we will  further assume that errors act independently on each qubit.

\begin{definition}[Memoryless Pauli Channel] \label{def-pauli-ch}
On a memoryless Pauli channel, the probability of error $E = \bigotimes_{i=1}^{n} E_{i} \in \bar \cG_n$, with $E_i \in \bar\cG_1$, is $\Prob(E) = \prod_{i=1}^{n}q_{i, E_{i}}$ where the $q_{i} $ are properly normalized probability vectors on $\{I,X,Y, Z\}$.
\end{definition}
In general, the probabilities $q_{i, E_{i}} $ can be different for all qubits. One important feature of memoryless Pauli channels is that they can be efficiently specified (to finite accuracy), i.e. using $\cO(n)$ bits of information. An obvious simplification of the above channel is made by supposing that the noise rates are the same for all $X,Y,Z$ type errors on each qubit. This specifies to a \emph{depolarizing channel}, and can be expressed as
\begin{gather}
\Prob(E) = \left(\dfrac{p}{3}\right)^{\wt(E)}\times (1-p)^{n - \wt(E)} . \label{eq-depol-err}
\end{gather}
Alternatively, one can assume that each qubit is first affected by  $X$-type errors,  and a then by $Z$-type errors, and moreover, these errors are independent of each other. A $Y$-type error occurs only when both a $Z-$type as well as an $X-$type error affect a qubit. This specifies a \emph{independent $X-Z$ channel}, which we will use to prove our main result.
\begin{definition}[Independent X-Z channel] \label{def-XZ-ch}
An independent $X-Z$ channel is a memoryless Pauli channel defined by $q_X = q_Z = \slfrac{p}{2}(1-\slfrac{p}{2})$, $q_Y = \slfrac{p^{2}}{4}$, and $q_I = (1-\slfrac{p}{2})^{2}$ for each qubit. The probability of an error $E\in\bar{\cG}_{n}$ can thus be written as 
\begin{gather}
\Prob(E) = \left(\dfrac{p}{2}\right)^{\wtt(E)}\times\left(1-\dfrac{p}{2}\right)^{2n - \wtt(E)} \label{eq-prob-XZ}
\end{gather}
where $\wtt(E)$ is the symplectic weight of $E$, see Sec.~({}\ref {sec-symplectic}{}).
\end{definition}
One key feature of both the depolarizing channel and the independent $X-Z$ channel is that the probability of an error depends only on its weight (either $\wt$ or $\wt_2$). As a consequence, evaluating the probability of a logical class can be done by counting the number of its elements of a given weight, which puts the problem in $\sharpP$. 
Notice also that $\Prob(E)$ is monotonically decreasing with its weight for $p\in[0,\slfrac{1}{2}]$, implying that in such a range, a minimum weight error has the maximum probability. We will often refer to $p$ to as the error-rate per qubit or the physical noise-rate.


\section{The decoding problem} \label{decoding-problem}

In this section we define the decoding problem more formally, and explain how it is affected by the existence of degenerate errors as defined in Sec.~({}\ref {sec-deg-err}{}).

The qubits are prepared in a code state $\ket\psi \in \cQ$ and are subject to the memoryless Pauli channel (def. \ref{def-pauli-ch}). The received state is $\ket\phi = E \ket\psi$ where $E \in \overline{\cG}_n$ is an unknown error chosen from the distribution $\Prob(E) $. \emph{Decoding} refers to the operation performed by the receiver in recovering the state $|\psi\rangle$ from the state $|\phi\rangle = E|\psi\rangle$. Since all Pauli operators square to the identity, it suffices for the receiver to determine $E$ and apply it to the system to recover the original state, i.e. $E\ket\phi = E^{2}\ket\psi = \ket\psi$.

The error, being an element of the Pauli group, can either commute or anti-commute with each of the stabilizer generators. Thus, upon measurement of each stabilizer generator, we obtain an eigenvalue $+1(-1)$ indicating that $E$ commutes (anti-commutes) with the stabilizer generator:
\begin{equation}
S_j\ket\phi = S_{j}\cdot E\ket\psi = \left\{
\begin{array}{llll}
E\cdot S_{j}\ket\psi &= E\ket\psi &= +\ket\phi & {\rm if\ } E\cdot S_{j} = S_{j}\cdot E \\
-E\cdot S_{j}\ket\psi &= E\ket\psi &= -\ket\phi & {\rm if\ } E\cdot S_{j} = -S_{j}\cdot E .
\end{array}
\right.
\end{equation}
Each $\pm 1$ measurement outcome $m_{j}$ is encoded into a bit $s_{j}$ such that $m_{j} = (-1)^{s_{j}}$. The outcomes of measuring all the check operators is encoded as a $(n-k)$ bit vector $\vec{s}$ called the \emph{error syndrome}.

\subsection{Non-degenerate decoding}
The decoding problem  consists in identifying $E$ conditioned on knowledge of the error syndrome. As in the classical case \cite{ber-ecc-84,mc-ber-78,barg97}, the  conceptually simplest strategy is to choose $E$ that has the highest probability amongst all errors consistent with the measured syndrome. This is called \emph{Quantum Maximum Likelihood Decoding} ($\QMLD$) \cite{heish-gall-11,kuo-lu-12,yueh-13}, and can be formulated mathematically as:
\begin{equation}
E_{\QMLD}(\vec s) = \argmax_{E\in \bar\cG_n} \Prob(E\thickspace |\thickspace \vec s)
\end{equation}
Using the decomposition Eq.~({}\ref {decomp-pauli-op}{}), we can view the error probability $\Prob(E)$ as a joint probability over the group $\cT$, $\cL$, and $\cS$ in a natural way:
\begin{equation}
\Prob(T,L,S) = \Prob(E = T\cdot L\cdot S), \quad {\rm with \ } T\in \cT,\ L\in \cL, \rm{ \  and \ } S\in \cS.
\end{equation}
Using the commutation relations given at Eq.~({}\ref {commutations}{}), it follows that the knowledge of $\vec s$ is equivalent to knowledge of $T$, since $T_j$ is the only element of this basis that anti-commutes with $S_j$. Thus, we have
$T_{\vec s} = \prod_{j} T_{j}^{s_{j}}$, and all errors in $T_{\vec{s}}\cdot\cN(\cS)$ are consistent with the syndrome  $\vec{s}$. Hence, at this stage a \emph{best guess} for the elements in $\mathcal{S}$ and $\mathcal{L}$ needs to be employed.  This involves finding $L\cdot S$ that \emph{maximizes the likelihood} of $E = T_{\vec{s}}\cdot L\cdot S$, implying an equivalent definition of  $\QMLD$: 
\begin{equation}
E_{\QMLD}(\vec s) = T_{\vec s} \cdot \argmax_{L\in \cL, S\in \cS} \Prob(L,S | T_{\vec s})
\label{maxQMLD}
\end{equation}
where the conditional probability is given by Bayes' rule $\Prob(L,S|T_{\vec{s}}) = \Prob(L,S,T_{\vec{s}})/\Prob(T_{\vec{s}})$ with the marginal defined as usual $\Prob(T_{\vec{s}}) = \sum_{L,S} \Prob(L,S,T_{\vec{s}})$. 

Informally speaking, $\QMLD$ addresses the  problem of determining the element of $\mathcal{L}\cdot \mathcal{S}$, whose probability is maximum, given an error rate and a syndrome. For the special cases of the depolarizing channel Eq.~({}\ref {eq-depol-err}{}) and the independent $X-Z$ channel Def.~({}\ref {def-XZ-ch}{}), the search for an operator with maximum probability is synonymous to the search for an operator of minimum weight (with two possible notions of weight, one over the $\overline{\cG}_{n}$ and the other over $\mathbb{Z}_{2}^{2n}$). Consequently, $\QMLD$ is also known as a \emph{minimum-weight decoder}.

One subtlety arises in case where the maximum probability is a close tie, as we cannot expect a decoder to discriminate probabilities within an arbitrary accuracy. Thus, we can define $\QMLD$ as the problem of identifying the optimal couple $L, S$, but we tolerate that it fails when more than one choice have probabilities that are within a small distance $\Delta$ from the optimal. The standard way of formalizing this notion is with a \emph{promise gap}, where we just assume in the definition of the problem that there is no close tie.

\begin{definition}[Quantum Maximum Likelihood Decoding]{$\QMLD$:} \label{defn-QMLD}\\
{\bf Input:} An $n$-qubit quantum code with stabilizer group $\cS$ specified by $n-k$ independent generators, a memoryless Pauli channel $\Prob(E)$,  an error syndrome $\vec s \in \mathbb{Z}_{2}^{n-k}$, and a promise gap $\Delta$. \\
{\bf Promise:} There exists a couple $S^{\star}\in \cS$ and $L^{\star} \in \cN(\cS)/\cS$ such that
\begin{equation}
\Prob(L^{\star},S^{\star}|T_{\vec s}) - \Prob(L,S|T_{\vec s}) \geq \Delta \Prob(L^{\star},S^{\star}|T_{\vec s}) \thickspace,\thickspace  \forall\thickspace (L,S)\neq (L^{\star},S^{\star}).
\end{equation}
{\bf Output:} $L^{\star}S^\star$.
\end{definition}

As mentioned above, for the depolarizing channel and the independent $X-Z$ channel $\QMLD$ is formally equivalent to a minimum-weight decoder, in which case the promise gap is irrelevant and can be set to $0$. 

\subsection{Degenerate decoding} \label{sec-deg-dec}

We will now explain how degeneracy changes the decoding strategy. As explained in Sec.~({}\ref {sec-deg-err}{}), errors can be classified into equivalence classes labelled by $L,T$, with all errors within a class having the same effect on the code and therefore all being correctable by the same operation. As a consequence, we see that $\QMLD$ is a suboptimal decoding strategy---in the sense that it does not reach the maximum probability of correctly decoding---because it fails to recognize the equivalence between degenerate errors. Instead of searching the most likely error, the optimal decoder seeks for the most likely equivalence class of errors, with the probability of a class of errors equal to the sum of the probability of the errors it contains. Since all errors in an equivalence class are related by an element of $\cS$ and their $\cT$ component is fixed by the syndrome, we can write the probability of a class conditioned on syndrome $\vec{s}$ as 
\begin{equation}
\Prob(L|\vec s) = \sum_{S\in \cS} \Prob(L,S|T_{\vec s}),\label{eq:PLS}
\end{equation}
where we use standard Bayesian calculus as above. The \emph{Degenerate Maximum Likelihood Decoding} ($\DQMLD$) problem can be formulated mathematically as determining an error in the most probable logical class, for a given syndrome, Eq.~({}\ref {eq:PLS}{}).
\begin{definition}[Degenerate Maximum Likelihood Decoding]{$\DQMLD$} \label{def-DQMLD}\\
{\bf Input:} An $n$-qubit quantum code with stabilizer group $\cS$ specified by $n-k$ independent generators, a memoryless Pauli channel $\Prob(E)$,  an error syndrome $\vec s \in \mathbb{Z}_{2}^{n-k}$, and a promise gap $\Delta$. \\
{\bf Promise:} There exists an $L^{\star} \in \cN(\cS)/\cS$ such that
\begin{equation} \label{eq-DQMLD-gap}
\Prob(L^*|T_{\vec s}) - \Prob(L|T_{\vec s}) \geq \Delta \Prob(L^*|T_{\vec s}) \thickspace,\thickspace  \forall L\neq L^*.
\end{equation}
{\bf Output:} $L^*$.
\end{definition}

Note that the promise gap can be expressed as a relative gap or an additive gap, the former being the right one in our setting. This is because the relative promise gap can be related to the failure probability of the code under optimal decoding. Consider a large promise gap $\Delta = 1-4^{-k}\epsilon$. Rewriting the promise as $P(L|T_{\vec s}) \leq 4^{-k} \epsilon P(L^*|T_{\vec s})$ and summing over all $L\neq L^*$ (of which there are $4^k-1$), we arrive at $P(L^*|T_{\vec s}) \geq 1-2\epsilon$, which simply says that the probability that the error that occurred is not equivalent to $L^*$---and hence that the decoder fails---is at most $2\epsilon$.

Note also that for a fixed $T$, the probabilities $\Prob(L,S|T)$ and $\Prob(L,S,T)$ differ only by a constant, so we can perform the optimization in Def.~({}\ref {def-DQMLD}{}) on the joint probability instead of the conditional probability. The sum appearing in the this probability Eq.~({}\ref {eq:PLS}{}), being over $2^{n-k}$ terms, forbids a  polynomial-time direct computation of its value. However, for i.i.d. Pauli channels such as the   depolarizing channel and the independent $X-Z$ channel, by grouping terms of equal weight in the sum, we can express the sum in Eq.~({}\ref {eq-logical-prob}{}) more succinctly. For the case of the independent $X-Z$ channel Def.~({}\ref {def-XZ-ch}{}), the above joint probability can be expressed as
\begin{gather}
\Prob(T_{\vec{s}},L,S) = \left(1-\dfrac{p}{2}\right)^{2n}\sum_{S\in\mathcal{S}}\tilde{p}^{\wtt(T_{\vec{s}} \cdot L\cdot S)} \label{eq-logical-prob} 
\end{gather}
where $\tilde{p} = \slfrac{p}{(2-p)} $. By grouping terms of equal weight in the sum, we arrive at a sum involving only $n+1$ terms
\begin{gather}
\Prob(T_{\vec{s}},L,S) = \left(1-\dfrac{p}{2}\right)^{2n} \sum_{i = 0}^{n}A_{i}(\vec{s},L) \thinspace \tilde{p}^{i}
\label{dqmld-we-coeff} \\
\text{where }A_{i}(\vec{s},L) = \left| \{S\in\mathcal{S}\thickspace : \thickspace \wtt(E = T_{\vec{s}}\cdot L\cdot S) = i\}\right| , \text{ and } \sum_{i=0}^{n}A_{i}(\vec{s},L) = 2^{n-k}. \label{we-coeff-sum}
\end{gather}
The coefficients $\{A_{i}(\vec{s},L)\}_{i=0}^{n}$ are called the {\em weight enumerators of the coset} associated to $\vec s$ and $L$. Note that the coset weight enumerators play a very important role in estimating the decoder performances of both $\QMLD$ as well as $\DQMLD$.

The sum in Eq.~({}\ref {dqmld-we-coeff}{}) is now over polynomially many terms unlike its previous form Eq.~({}\ref {eq-logical-prob}{}). Computing such a sum for each logical operator and subsequently optimizing over their values would solve $\DQMLD$. An $[[n,k]]$ stabilizer code has $|\cL| = 4^{k}$, implying that even if the weight enumerators can be computed efficiently, a polynomial-time optimization cannot be performed over the different cosets labeled by $L$, unless $k \in \cO(\log(n))$, which is the regime that we are interested in. Furthermore, we believe $\DQMLD$ is at least as hard otherwise, i.e for $k \in \Omega(\log(n))$.

At this stage, we like to remark that though $\QMLD$ and $\DQMLD$ are stated as decision problems in \cite{heish-gall-11,kuo-lu-12}, the problem of practical interest is one of determining a Pauli operator that maximizes the respective probabilities in Defs.\nobreakspace(\ref{defn-QMLD}) and (\ref {def-DQMLD}). To enable a decision problem formulation in \cite{heish-gall-11,kuo-lu-12}, both $\QMLD$ and $\DQMLD$ are defined to take as input an additional constant $c$. Subsequently, $\DQMLD$ would be made to decide the existence of an $L\in\cL$ whose probability Eq.~({}\ref {eq-logical-prob}{}) is at least $c$. This appears to be achieving more than what is really expected from of $\DQMLD$ (or $\QMLD$) in practice. In particular, by varying this constant $c$, along with the input in Def.~({}\ref {def-DQMLD}{}), one could use the oracle to not only learn the optimal correction but also the probability of its equivalence class. The latter is not necessary to perform error correction and thus gives more power to the decoder oracle than it should. This is why we formulated $\DQMLD$ as a function problem that does not explicitly reveal the probability of the equivalence class of the optimal correction, and therefore we consider it to be closer to the real world decoding problem.

\subsection{Importance of degeneracy}

Before addressing the computational complexity of degenerate decoding, we close this section with a discussion of its practical relevance. The two decoders $\QMLD$ and $\DQMLD$ will provide different answers whenever the most likely equivalence class does not contain the error with the largest probability. Consider the hypothetical scenario where the class of error $L_1$ contains a single error of low weight $a$ and $2^{n-k}-1$ errors of high weight $b \gg a$, and that the class $L_2$ contains $2^{n-k}$ errors of intermediate weight c, $a\ll c\ll b$. The $\QMLD$ would chose the error from the class $L_1$, because it is the most likely error. On the other hand, the probabilities of these two classes are given by
\begin{equation}
\Prob(L_1) \propto \tilde  p^a + (2^{n-k}-1) \tilde p^b \approx \tilde p^a \ {\rm and} \ P(L_2) \propto 2^{n-k} \tilde p^c.
\end{equation}
Thus, we see that $\DQMLD$ would provide a different answer if $P(L_1) < P(L_2)$, or equivalently $\tilde p > 2^{-\frac{n-k}{c-a}}$, and that the two decoders would agree otherwise. Thus, for sufficiently low error rates, and in particular in the extreme limit $p \leq 2^{k-n}$, degeneracy does not affect the decoding problem. 

More importantly, degeneracy becomes unimportant when the {\em failure rate} of the code is very low, which doesn't necessarily require a low physical noise rate. Remember Sec.~({}\ref {sec-deg-dec}{}) that the failure rate $2\epsilon$ of the code is related to the $\DQMLD$ promise gap $\Delta = 1-4^{-k}\epsilon$. The following lemma, whose proof is presented in App.~({}\ref {app-large-gap}{}), shows that for $\epsilon = 2^{k-n}$, $\QMLD$ provides the same answer as $\DQMLD$, which in turn implies that $\DQMLD$ is in $\NP$ with such a large gap.
\begin{lemma} \label{lemma-large-gap}
With a promise gap $\Delta = 1-2^{-n-k}$ and on an independent $X-Z$ channel, $\DQMLD$ and $\QMLD$ are equivalent.
\end{lemma}

In the light of these observations, one might imagine that in general, $\DQMLD$ can at most offer a marginal improvement over $\QMLD$, i.e. that for a given code and noise rate, the decoding error probability of $\QMLD$ is upper bounded by a function of the decoding error probability of $\DQMLD$. There is strong evidence that this is not the case however. Monte Carlo simulations \cite{HPP01a,WHP03a,BAOK12a} have shown that $\QMLD$ and $\DQMLD$ achieve different error thresholds with Kitaev's topological code. (This statement is equivalent to the fact that the critical disorder strength separating the ordered from the disordered phase in the random bound Ising model decreases below the Nishimori temperature.) Thus, for noise rates $p$ falling between these two thresholds, the failure probability of $\DQMLD$ tends to 0 as the number of qubit $n$ increases, while the failure probability of $\QMLD$ tends to $1-4^{-k}$ (the failure probability of a random guess), so the performances of both decoders can be significantly different.

Degeneracy can also severely impact the performances of certain decoding algorithms. Due to degeneracy of quantum errors, the ability to correct errors does not imply that the a posteriori marginal error probability over individual qubits is sharply peaked, in contrast to the classical setting.  This is the case in particular for low density parity check (LDPC) codes. These codes  have the property of admitting a set of stabilizer generators $S_j$ of weight bounded by a constant. As a consequence, the weight of equivalent errors $E$ and $E'$ related by a stabilizer generator $E' = ES_j$ will differ at most by a small constant. Thus, we expect degeneracy to play an important role: each equivalence class contains many errors of roughly the same weights. As discussed in \cite{david-08}, conventional decoding algorithms for LDPC codes (belief propagation \cite{AM00a}) are {\em marginal decoders} in the sense that they optimize the probability of error independently for each qubit. But in the presence of degeneracy, we can have a probability sharply peaked over a single equivalence class of errors---ensuring the success of $\DQMLD$---but yet have a very broad marginal distribution over individual qubits---leading to a failure of a marginal decoder. So the existence of a good, general purpose decoder for quantum LDPC codes, playing a role analogous to belief propagation in the classical setting, remains an outstanding open question. 

This situation is best illustrated with Kitaev's topological code \cite{Kit03a}. In this code, errors correspond to strings on a regular square lattice, and error syndromes are located at the endpoints of the strings. The weight of an error is equal to the length of the corresponding string. Lastly, the different equivalence classes of errors correspond to the homology classes of the lattice. For a syndrome configuration shown at Fig.~({}\ref {fig-toric-code-marginal}{}), all the short paths have the same homology, so the probability is sharply peaked over one equivalence class. But there are several distinct strings of the same length compatible with this syndrome, so the marginal error probability over individual qubits is very broad.
\begin{figure}[tbh]
\centering
\subfloat[]{\label{fig-toric-code-syndromes}\includegraphics[scale=0.45]{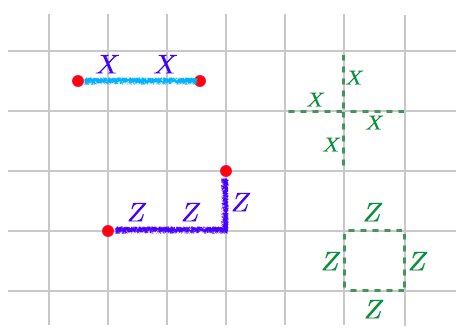}}
\hspace{0.4cm}
\subfloat[]{\label{fig-toric-code-marginal}\includegraphics[scale=0.4]{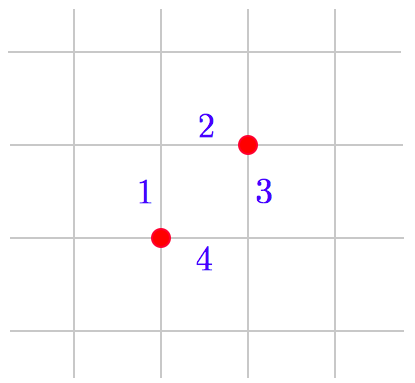}}
\label{fig-toric-code-errors}
\caption{The toric code has two type of stabilizer generators, shown here in green Fig.~({}\ref {fig-toric-code-syndromes}{}), centered around each face and each vertex of the lattice. Strings of $Z$ errors on the lattice generate vertex syndromes at their endpoints, while strings of $X$ errors on the dual lattice generate face syndromes at their endpoints. Strings with same holonomy represent equivalent errors. In Fig.~({}\ref {fig-toric-code-marginal}{}) the qubits are labelled $1-4$ for easy reference. There are many errors consistent with this syndrome, the lowest weight of which are $Z_{1}Z_{2}$ and $Z_{3}Z_{4}$, and by symmetry $\Prob(Z_{1}Z_{2}) = \Prob(Z_{3}Z_{4}) = \max_{E\in \overline{\cG}_{n}}\Prob(E |\vec s)$. Applying either of these operators serves as a valid correction since they have the same holonomy. Errors with a different holonomy have a probability that decreases exponentially with the lattice size, so $\DQMLD$ would successfully decode. But marginally, qubits 1, 2, 3, and 4 all have identical probabilities, so a marginal decoder would either correct with $Z_1Z_2Z_3Z_4$ or the identity, none of which are acceptable corrections.}
\end{figure}

Lastly, we note that to achieve the true capacity of certain quantum channels (as opposed to the single-shot capacity), it is necessary to encode the information in a degenerate code \cite{SS07b}. In other words, there are certain channels that could not be used to send any quantum information if a non-degenerate code was used, but can reliably do so at a finite rate with degenerate codes. We do not know however if degeneracy also needs to be taken into account during the decoding process to realize this. In particular, the example in \cite{SS07b} uses a generalization of Shor's code \cite{Sho95a}, for which $\QMLD$ and $\DQMLD$ always yield the same output. We know of only a few examples of codes for which $\DQMLD$ can be computed efficiently, namely concatenated codes \cite{poulin-06}, convolutional codes \cite{dav-emi-12}, and Bacon-Shor codes \cite{preskill-13}. There exist heuristic methods to take degeneracy into account in topological codes \cite{dav-gui-10} and turbo codes \cite{poulin-tillich-09}.

\section{Complexity of the decoding problem} \label{sec-hardness}

The one-to-one correspondence between $[n,k]$ stabilizer codes and $(2n,k)$ symplectic linear codes \cite{got-phd-97} is used in \cite{heish-gall-11,kuo-lu-12} to show that a solution for $\QMLD$ can be used to decide an $\NPComplete$ problem in polynomial time. Consequently, $\QMLD\in\NPComplete$.

$\DQMLD$ was shown to be $\NPHard$ for the case of an independent $X-Z$ channel in \cite{heish-gall-11} and depolarizing channel in \cite{kuo-lu-12}, using a reduction from a $\NPComplete$ problem pertaining to classical linear code. We now state our main result, which establishes that $\DQMLD$ is in fact much harder than what there previous results anticipated.
\begin{theorem}\label{main-theorem}
$\DQMLD$ (c.f. Def.~({}\ref {def-DQMLD}{})) of an $[[n,k=1]]$ stabilizer code on an independent $X-Z$ channel and with a promise gap $\Delta \leq 2[2 + n^{\lambda}]^{-1}$, with $\lambda = \Omega(\textsf{polylog}(n))$, is in $\sharpPComplete$.
\end{theorem}
In the following sections, we will show that for a classical binary linear code $\cC$ and $\lambda \in [0,n]$, the problem of computing $\WE_{i}(\cC)$ for $i=0,1,2,\ldots, \lambda$ is polynomial time Turing reducible to $\DQMLD$ on an independent $X-Z$ channel, for $\Delta \leq 2[2 + n^{\lambda}]^{-1}$.

In the case of a general memoryless Pauli channel, it is not always possible to express the probability of a logical class as a weight enumerator sum Eq.~({}\ref {dqmld-we-coeff}{}) with polynomially many terms in the sum. Hence, in such cases, the containment of $\DQMLD$ in $\sharpP$ is not known and we can only claim that it is $\sharpPHard$. However, whenever one can express $\Prob(L|\vec{s})$ as a sum with polynomially many terms, each of which is a $\sharpP$ function, then $\DQMLD$ can be put into $\sharpP$. 

For the independent $X-Z$ channel, the containment in $\sharpP$ is straightforward.
\begin{lemma} \label{dqmld-sharpP}
$\DQMLD$ in Def.~({}\ref {def-DQMLD}{}) on a $[[n,k]]$ stabilizer code with $k = \Omega(\log_{2}n)$ on an independent $X-Z$ channel Eq.~({}\ref {eq-prob-XZ}{}) or a depolarizing channel, is in $\sharpP$.
\end{lemma}
\begin{proof}
From the definition of $\DQMLD$, one can subdivide it into two problems. The first step consists of computing the probabilities of all the logical classes, corresponding to the given syndrome, followed next by an optimization over the probabilities, thereby choosing an error from the logical class with the largest probability. Since $k = \Omega(\log_{2}n)$, there are at most polynomially many logical classes and hence a naive search for the maximum amongst the set of probabilities of all the logical classes can be achieved in polynomial-time.

It now remains to show that for each logical class, its probability Eq.~({}\ref {we-coeff-sum}{}) for a general noise rate $p$, is a $\sharpP$ function. This immediately follows from the polynomial form in $p$, taken by the logical class probabilities, Eq.~({}\ref {dqmld-we-coeff}{}) and the observation that the the coefficients of the polynomial are weight enumerator coefficients of a suitable linear code Sec.~({}\ref {sec-symplectic}{}).

The promise gap $\Delta$ is irrelevant to the containment of $\DQMLD$ in $\sharpP$ unless $\Delta = o(2^{-\poly(n)})$. A gap of $2^{-\poly(n)}$ is essential since the probabilities of various logical classes are represented as bit-strings for performing arithmetic and they must differ in the first polynomially many bits.
\end{proof}

\section{Reduction} \label{sec-red-over}
In this section, we present a polynomial time algorithm that accepts as input a classical linear code $\mathcal{C}$ and outputs $\{\WE_{i}(\mathcal{C})\}_{i=0}^{n}$ as in Def.~({}\ref {WE-def-code}{}) by querying a $\DQMLD$ oracle with an independent $X-Z$ noise model and a promise gap $\Delta$ which is $\slfrac{1}{\textsf{quasi-polynomial}(n)}$.

\subsection{Reduction overview}
The correspondence between symplectic linear codes and stabilizer codes discussed in Sec.~({}\ref {sec-symplectic}{}) implies that with an independent $X-Z$ channel, the probability of a logical class is related to the weight enumerator polynomial of a classical linear code. In particular, with the trivial syndrome $\vec s = \vec 0$ and at very low noise rate $p\ll \slfrac{1}{2}$, the most likely equivalence class is always the trivial one $L = \mathbb{I}$. Hence, in this setting, the probability of the trivial logical class for a quantum code with stabilizers $\cS$ can be expressed, similar to Eq.~({}\ref {dqmld-we-coeff}{}), in terms of a corresponding classical code $\cC_{\cS}$ as follows 
\begin{gather}
\Prob(\mathbb{I}\thinspace|\thinspace\vec{0}) = \left(1-\dfrac{p}{2}\right)^{2n}\sum_{i=0}^{2n}\WE_{i}(\cC_{\cS})\tilde{p}^{i} \label{prob-I-WE}
\end{gather}

Since this probability is a polynomial of degree $2n$, determining its value for $2n+1$ distinct values of the physical noise rate $p$ would enable us to determine the weight enumerator of the corresponding classical code. There are two caveats to this approach. First, this approach only works for classical codes $\cC_{\cS}$ whose generator matrix corresponds to the symplectic representation of a quantum stabilizer code. Second, this approach requires knowledge of the probability of a certain equivalence class, while $\DQMLD$ only outputs the equivalence class with the largest probability; it does not reveal the value of the corresponding probability. 

The first caveat can easily be circumvented, for instance by padding the $n$-bit classical code with $n$  additional $0$s, thus obtaining a valid symplectic representation of a $n$-qubit quantum code (one whose stabilizer generators contain only $Z$ operators). Consequently, all stabilizers have weights between 0 and $n$ and the probability of the trivial logical class can be expressed as Eq. (\ref{prob-I-WE}), where the range of sum is up to $n$. To circumvent the second caveat, we need to use the $\DQMLD$ oracle to obtain equality constraints on the weight enumerator of $\cC$. This is done by varying the physical noise rate $p$, always keeping the syndrome trivial.  As mentioned above, at very low noise rate the optimal logical class is $\mathbb I$. Increasing the noise rate $p$, we will reach a \emph{crossing point} $p_{1}$ where the $\DQMLD$ output changes from $\mathbb{I}$ to $L^{\star} \neq \mathbb{I}$. At this point, the promise gap condition is violated, i.e, $\left|\Prob(\mathbb{I}|\vec s) - \Prob(L^{\star}|\vec s)\right| \leq \Delta \Prob(\mathbb{I}|\vec{s})$, which can be expressed in terms of weight enumerators as:
\begin{equation}
\sum_{i=0}^{n} \WE_{i}(\cC)\tilde{p}_{1}^{i} - \sum_{i=0}^{n} B_{i} \tilde{p}_{1}^{i} \leq \Delta \sum_{i=0}^{n}  \WE_{i}(\cC)\tilde{p}_{1}^{i} \label{eq-crossing-I-L}
\end{equation}
where $B_{i}$ are the weight enumerators of an affine code. In the case where $\Delta = 0$, this crossing point provides an equality condition between two polynomials, which is what we are seeking. But since these are polynomials with integer coefficients, knowing the location of a crossing point within a finite accuracy, which translate into a finite promise gap $\Delta$, is enough to determine the exact crossing point, see Lemma.~({}\ref {gap-size-qpoly}{}). As we will show, there exists a fixed range of $p$ that provably contains a unique crossing point, enabling a polynomial-time accurate determination of the crossing point using a binary search procedure.

This gives us one potential equality, but introduces more unknown coefficients $B_{i}$. To get additional linear constraints, we modify the code in a very special way, described in Sec.~({}\ref {sec-stab-cons}{}). This modification requires adding one {\em tunable} qubit and one stabilizer generator. By varying the noise rate on the tunable qubit over a range of values, we can change the location of the crossing point, and thus obtain new linear constraints relating $\{\WE_{i}(\cC)\}_{i=0}^{n}$ and the $\{B_{i}\}_{i=0}^{n}$. Repeating this procedure $2n+2$ times and making sure that all the linear constraints are linearly independent Sec.~({}\ref {sec-indep-cons}{}) enable us to determine the weight enumerator coefficients.  While the ability to change the noise rate of the tunable qubit gives us more linear constraints, it breaks the requirement that the noise model be the independent $X-Z$ channel with the same strength on all qubits. We will fix this problem in App.~({}\ref {app-concatenate}{}) by showing that the required channel can be simulated by concatenating the code with a Shor code. In fact, we will use this technique repeatedly in our proof.

\subsection{Stabilizer code construction} \label{sec-stab-cons}
Let $G$ be the $k\times n$ generator matrix of an $(n,k)$ classical linear code $\cC = \{x\in {\mathbb Z}_2^n : x = yG, \ y\in {\mathbb Z}_2^k \}$. Denote $\{g_i\}_{i=1,\ldots k}$ the rows of $G$ and let $\{g_i\}_{i=k+1\ldots n}$ be a generating set of the complement of the row space of $G$, i.e. in such a way that $\{g_i\}_{i=1\ldots n}$ span $\mathbb Z_2^n$. Construct a matrix $\tilde G $ with rows$\{g_i\}_{i=1\ldots n}$. This matrix is full rank, and therefore has an inverse $H$ that can be computed efficiently, and obeys $\tilde GH^{T} = \mathbb I$. Denote the rows of $H$ by $\{h_i\}_{i=1}^{n}$. 
 
We define a  $[[2n-k-1,1]]$ quantum code with stabilizer generators and logical operators  given by
\begin{align}
S_i &= Z^{g_i} & {\rm for }\ i=1,\ldots,k \\
S_{k+i} &= Z^{g_{k+i}} \otimes Z_{n+i} & {\rm for }\ i=1,\ldots,n-k-1 \\
S_{n-1+i} &= X^{h_{k+i}} \otimes X_{n+i} & {\rm for }\ i=1,\ldots,n-k-1 \\
\overline Z &= Z^{ g_n} &\\
\overline X &= X^{h_n}, &
\end{align}
where it is implicitly assumed that operators are padded to the right by identities to be elements of $\cG_{n+k-1}$. The validity of the resulting code can be verified from the fact that 
\begin{enumerate}
\item There are in total $2n-k-1$ qubits.
\item The $2n-k-2$ stabilizer generators are independent. This follows from the linear independence of the $h_i$ and the linear independence of the $g_i$, together with the fact that $X$-type operators are linearly independent of $Z$ type generators.
\item The stabilizer generators mutually commute. This is trivial among the first $n-1$ generators as they contain only $Z$ operators and similarly among the last $n-k-1$ last generators. Between these two sets, the commutation follows from the fact that $h_i \cdot g_j = 0$ except when $i=j$, in which case the presence of additional $X$ and $Z$ on the $n+i$th qubit ensures commutation. 
\item The logical operators commute with the stabilizer generators. This follows from the fact that $h_i\cdot g_j = 0$ for $i\neq j$, and the fact that $X$-type operators commute among themselves and similarly for $Z$-type operators.
\end{enumerate}

As discussed in Sec.~({}\ref {sec-deg-dec}{}), the probability of the trivial logical class $\mathbb I$ given a trivial syndrome $\vec s = \vec 0$ is simply the sum of the probabilities of all stabilizer group elements. Suppose now that the last $n-k-1$ qubits are error-free, while the other qubits are subject to an independent $X-Z$ channel. Then, the probability of an element of $\cS$ is zero if it contains a generator $S_i$ from the above list with $i>k$. Otherwise, this element of $\cS$ can be written as $S = Z^x \otimes {\mathbb I}$ for some $x\in {\mathbb Z}_2^n$ and its probability is $(\slfrac{p}{2})^{|x|}(1-\slfrac{p}{2})^{2n-|x|}$. We conclude that the probability of the trivial logical class is given by Eq.~({}\ref {prob-I-WE}{}) with $\cC$ the classical code defined by the generating matrix $G$.

Constraints on the weight enumerator polynomial will be obtained by finding crossing points where $\Prob({\mathbb I}|\vec 0) \approx \Prob(L|\vec 0)$ with $L \neq {\mathbb I}$. For technical reasons, we would like to be able to choose which $L$ will be the one realizing the crossing. This is because we want to force the crossing to happen with the same $L$ every time. To do this, we will modify the stabilizer by adding an extra qubit and making the transformations
\begin{align}
S_i &\rightarrow S_i\otimes {\mathbb I} \\
\overline Z &\rightarrow \overline Z \otimes {\mathbb I} \\
\overline X &\rightarrow \overline X \otimes  X 
\end{align}
to the stabilizers and logical operators, and adding the following stabilizer generator
\begin{equation}
S_{2n-k} = \overline Z \otimes Z.
\end{equation}
This defines an $[[2n-k,1]]$ stabilizer code, and its validity can easily be verified given the commutation relations worked out above. Moreover, if we assume that the added $(2n-k)$th qubit is also error-free, then the probability of the trivial logical class $\mathbb I$ given a trivial syndrome $\vec s = \vec 0$ is unchanged, and moreover the only other logical class with non-zero probability is the one associated to $\overline Z$, i.e. $\Prob(\overline X|\vec 0) = \Prob(\overline Y|\vec 0) = 0$.  

We need to perform one last modification to the code in order to be able to tune the crossing point, and hence obtain linearly independent equalities between $\Prob({\mathbb I}|\vec 0)$ and $\Prob(\overline Z|\vec 0)$. This transformation is quite similar to the previous one, and given by
\begin{align}
S_i &\rightarrow S_i\otimes {\mathbb I} \\
\overline Z &\rightarrow \overline Z \otimes Z \\
\overline X &\rightarrow \overline X \otimes  \mathbb I 
\end{align}
and adding the following stabilizer generator
\begin{equation}
S_{2n-k+1} = \overline X \otimes X.
\end{equation}
For this last qubit, we will assume a noise model where $p_X = p_Y = 0$, $p_Z = q$, and $p_{\mathbb I} = 1-q$ with $q$ being a tunable parameter. With this last choice, the only two non-zero probabilities of logical class conditioned on the trivial syndrome are given by
\begin{align}
\Prob({\mathbb I}|\vec 0) &= \frac 1\cZ (1-q) \sum_{i=0}^{n}  \WE_{i}(\cC)\tilde{p}_{1}^{i} \label{stab-cons-I} \\
\Prob(\overline Z|\vec 0) &= \frac1\cZ q\sum_{i=0}^{n} B_{i} \tilde{p}_{1}^{i} \label{stab-cons-Z}, 
\end{align}
where $\tilde{p} = \slfrac{p}{(2-p)}$ as above, $\cZ$ is a suitable normalization factor, and $B_i$ are the weight enumerators of the affine code associated to the $\overline Z$ logical class
\begin{equation}
B_i = \big| \{ x \in  \cC+g_n :  |x| = i\} \big|.
\end{equation}
A crossing point is observed when
\begin{equation}
v \sum_{i=0}^{n}  \WE_{i}(\cC)\left(\dfrac{\tilde{p}_{1}}{2}\right)^{i} - \sum_{i=0}^{n} B_{i} \left(\dfrac{\tilde{p}_{1}}{2}\right)^{i} \leq \Delta v\sum_{i=0}^{n}  \WE_{i}(\cC)\left(\dfrac{\tilde{p}_{1}}{2}\right)^{i} \label{CPv}
\end{equation}
where $v = (1-q)/q$ is a tunable parameter over the positive reals. Changing the value of $v$ will change the crossing point between these two logical classes, and provide linear constraints between two degree $n$ polynomials. If we can identify $2n+1$ such crossing points, it would provide enough information to retrieve the two polynomials, and hence solve the weight enumerator problem.

\subsection{Finding crossing points} \label{sec-crossing}
At this point, we have a deterministic procedure that, given any classical linear code $\cC$, can be used to generate linear constrains on its weight enumerator coefficients. Clearly, the overhead in the runtime of this procedure is the time required to spot a crossing point. A crossing point can potentially be observed at a physical noise rate $p$ anywhere between 0 and 1. One obvious indication of a crossing point is the switch in the output of the $\DQMLD$ oracle as we move $p$ across the crossing point. However if we move $p$ across two crossing points, we will not notice any net switch in the outputs of the $\DQMLD$ oracle. For this reason, we now want to restrict the values of $p$ to a range where we can prove that there is at most one crossing point. This will be possible by restricting the tunable parameter $v$ to a small interval near 0.

\begin{lemma} \label{gap-size-qpoly}
Given the stabilizer code defined above with a tunable parameter $v \leq(1-\Delta)\thinspace n^{-d}$, where $1 \leq d \leq |g_{n}|\leq n$ is the distance between $\cC$ and $\cC+g_n$, then there exists exactly one crossing point between the pair of logical classes $\mathbb{I}$ and $\overline{Z}$ in the interval $0\leq p\leq \slfrac{1}{n}$.
\end{lemma}
\begin{proof}
First, note that for sufficiently small $p$ and for any value of $v$ , $\DQMLD$ will output the identity class, simply because its probability is a polynomial in $p$ with a constant term, and the probabilities of all other logical classes contain no constant term. Furthermore, we claim that the probability of the trivial logical class is a strictly decreasing function of the noise rate $p$, when $0\leq p \leq \slfrac{1}{n}$. To justify this, it suffices to show that the derivative of $\Prob(\mathbb{I}|\vec{0})$ is strictly negative in the prescribed range of noise rates. Recalling the weight enumerator polynomial for $\Prob(\mathbb{I}|\vec{0})$ and computing the derivative w.r.t $p$, we find:
\begin{flalign}
\dfrac{\partial \Prob(\mathbb{I}|\vec{0})}{\partial p} 
&= \dfrac{1}{2}\sum_{i = 0}^{n}\WE_{i}(\cC)\left(\frac{p}{2}\right)^{i}\left(1-\frac{p}{2}\right)^{2n-i}\left[\dfrac{i}{\slfrac{p}{2}} - \dfrac{2n-i}{1-\slfrac{p}{2}}\right] \\
&= \sum_{i=1}^{n}\WE_{i}(\cC)\left(\frac{p}{2}\right)^{i}\left(1-\frac{p}{2}\right)^{2n-i}\left[\dfrac{i-np}{p(1-\slfrac{p}{2})}\right] - n\left(1-\frac{p}{2}\right)^{2n-1} \label{eq:rand}
\end{flalign}
where the last term is added because we have changed the range of summation. It only remains to show that the expression above is strictly negative. Clearly, this cannot be true for all $p \in [0,1]$. However, when $p\leq \slfrac{1}{n}$, we know that the first term in Eq.~({}\ref {eq:rand}{}) is non-negative. Indeed, all of its terms are positive by definition, except the one in square bracket which is non-negative when $p\leq \slfrac{1}{n}$. We claim now that when $p\leq \slfrac{1}{n}$, the second term of Eq.~({}\ref {eq:rand}{}) is negative and greater in norm than the first one, so the entire expression is strictly negative. This can be observed by the following inequality:
\begin{flalign*}
\sum_{i=1}^{n}\WE_{i}(\cC)\left(\frac{p}{2}\right)^{i}\left(1-\frac{p}{2}\right)^{2n-i}&\left[\dfrac{i-np}{p(1-\slfrac{p}{2})}\right] < \sum_{i=1}^{2n}\binom{2n}{i}\left(\frac{p}{2}\right)^{i}\left(1-\frac{p}{2}\right)^{2n-i}\left[\dfrac{i-np}{p(1-\slfrac{p}{2})}\right] \\
&= \sum_{i=0}^{2n}\binom{2n}{i}\left(\frac{p}{2}\right)^{i}\left(1-\frac{p}{2}\right)^{2n-i}\left[\dfrac{i-np}{p(1-\slfrac{p}{2})}\right] + n\left(1-\frac{p}{2}\right)^{2n-1} \\
&= n\left(1-\frac{p}{2}\right)^{2n-1}.
\end{flalign*}
Hence, we see that indeed the probability of the trivial logical class is strictly decreasing when $0\leq p \leq \slfrac{1}{n}$. Since there are only two logical classes in our setting, it is clear that the probability associated to the other class is increasing in that interval. 

We have identified an interval where the probabilities are monotonic, and what remains to be shown is that there is indeed a crossing point inside this interval when the parameter $v$ is chosen carefully. Intuitively, we can see that decreasing the value of $v$ (and hence of the trivial logical class) will decrease the value of $p$ where the first crossing point occurs.  We will now set an upper bound $p_{\max}$ on the value of $p$ where the first crossing point occurs. The first crossing point will occur at the latest when $v\Prob(\mathbb{I}|\vec{0}) - \Prob(\overline{Z}|\vec{0}) = \Delta \Prob(\mathbb{I}|\vec{0})$, so equivalently $v(1-\Delta) = \slfrac{\Prob(\overline{Z} | \vec{0})}{\Prob(\mathbb{I}|\vec{0})}$. On the other hand, the ratio $\slfrac{\Prob(\overline{Z} | \vec{0})}{\Prob(\mathbb{I}|\vec{0})}$ is lower-bounded by $[p/(2-p)]^d$ since each word in $x\in \cC$ is mapped onto a word $y=x+g_{n}$ of weight at most $|x|+|g_{n}|$ in $\cC+g_{n}$. Hence, the first crossing occurs for a value of $p$ lower or equal to the point where
\begin{equation}
v = (1-\Delta)^{-1} \left[\dfrac{p}{2-p}\right]^{d}, \ {\rm or\ equivalently\ } p_{\max} =  \dfrac{2}{1 + (1 - \Delta)^{\frac{1}{d}}\thickspace v^{-\frac{1}{d}}}.
\end{equation}
By choosing 
\begin{gather}
p_{\max} \leq \dfrac{1}{n}, \thickspace  \ {\rm or\ equivalently\ } \thickspace v\leq (1-\Delta)^{-1}\thickspace n^{-d}, \label{eq-max-crossing}
\end{gather}
we are sure that the first crossing point occurs in the monotonic region, and hence that the interval $0\leq p \leq \slfrac{1}{n}$ contains a single crossing point.
\end{proof}

The existence of a unique crossing point in the interval $p\in(0,\slfrac{1}{n}]$  enables a \emph{binary-search like} procedure to narrow down on the possible location of a crossing point. If $\DQMLD$ produces the same output for a pair of $p$'s in that interval, it implies that no such crossing point exists between them, and furthermore, that a crossing point lies outside of this interval. This halves the size of the interval between the next pair of points to be queried with. Hence the location of the crossing point can be obtained to an accuracy of $2^{-j}$ with at most $j$ queries to $\DQMLD$. 

\subsection{Independent constraints and the promise gap} \label{sec-indep-cons}
We have described a polynomial-time method to estimate the location of crossing points within exponential accuracy, i.e. values of $p$ where Eq.~({}\ref {CPv}{}) is fulfilled. It remains to show that a polynomial number of such  linear constraints are sufficient to determine the weight enumerators of the linear code. Every crossing point provides a linear constraint on $\{\WE_{i}(\cC),B_{i}\}_{i=0}^{n}$. However these conditions are inequalities.  In the following lemma, we establish that for $\Delta$ sufficiently small, the system of inequalities provides the same integer solutions as the system of equalities (i.e., obtained with $\Delta = 0$).
\begin{lemma} \label{lemma-gap-size}
Provided that $\Delta \leq 2[2 + n^{\lambda}]^{-1}$, the first $\lambda$ weight enumerator coefficients $\{\WE_{i}(\cC)\}_{i=0}^{\lambda}$ can be extracted efficiently from the value of $2n+1$ crossing points $p_k$ that produce linearly independent inequalities Eq.~({}\ref {CPv}{}). 
\end{lemma}
\begin{proof}
Given a crossing point $p_k$, we can rewrite the inequality Eq.~({}\ref {CPv}{}) as an equality
\begin{gather*}
-v_k\thickspace \sum_{i=0}^{n}\WE_{i}(\cC)\tilde{p_k}^{i} + \sum_{i=0}^{n}B_{i}\tilde{p_k}^{i} + \delta_k\Delta\left[v_k\thickspace \sum_{i=0}^{n}\WE_{i}(\cC)\tilde{p_k}^{i}\right] = 0
\end{gather*}
by introducing an unknown parameter $0\leq\delta_k\leq 1$. 
Given $2n+1$ distinct crossing points $p_k$, we obtain a linear system $\sfM\cdot\omega - \Delta\thickspace (\cJ\cdot\sfM)\cdot\omega = \vec{0}$, where $||\cJ||_{\infty} \leq 1$ and $\omega$ is the vector containing the weight enumerator coefficients. Since the location of a crossing point is invariant under multiplying both weight enumerators $\{\WE(\cC)_{i},B_{i}\}_{i=0}^{n}$ by a constant, it is clear that $2n+1$ linear independent constraints alone would result in the trivial solution. However, the normalization condition
\begin{gather}
\sum_{i=0}^{n}(\WE(\cC)_{i} + B_{i}) = 2^{n} \label{norm-WE-coeffs}
\end{gather}
along with the $2n+1$ linearly independent constraints obtained from the location of crossing points is sufficient to determine the weight enumerators coefficients in $\vec{\omega}$, 
\begin{gather}
\vec{\omega} \thickspace = \thickspace \left[\thinspace\left[\mathbb{I} + \Delta\thickspace \cJ\right]\cdot\sfM \thinspace\right]^{-1}\cdot \vec{b} \label{eq-M-inv-gap}
\end{gather}
where we have defined
\begin{gather}
\sfM = \left[\begin{array}{cccccccc}
1 & \tilde{p}_{1} & \dots & \tilde{p}_{1}^{n} & -v_{1} & -v_{1}\tilde{p}_{1} & \dots & -v_{1}\tilde{p}_{1}^{n} \\
  &  &  &  &  &  & \\
1 & \tilde{p}_{2} & \dots & \tilde{p}_{2}^{n} & -v_{2} & -v_{2}\tilde{p}_{2} & \dots & -v_{2}\tilde{p}_{2}^{n} \\
  &  &  &  &  &  & \\
\vdots &  & \ddots & \vdots & \vdots & & \ddots & \vdots \\
  &  &  &  &  &  & \\
1 & \tilde{p}_{2n+1} & \dots & \tilde{p}_{2n+1}^{n} & -v_{2n+1} & -v_{2n+1}\tilde{p}_{2n+1} & \dots & -v_{2n+1}\tilde{p}_{2n+1}^{n} \\
  &  &  &  &  &  & \\
1 & 1 & \dots & 1 & 1 & 1 & \dots & 1 \\
\end{array}\right] \label{constraint-matrix} \\
\vec{b} = \begin{pmatrix}0 & \cdots 0 & 2^{n}\end{pmatrix}^{\sfT} \textsf{ , } \vec{\omega} = \begin{pmatrix} B_{0} & \cdots & B_{n} & \WE_{0}(\cC) & \cdots & \WE_{n}(\cC)\end{pmatrix}^{\sfT} \label{weight-enumerator-coefficients-vector-b} \\
 \cJ = \left[\begin{array}{ccc|ccc}
0 & & & & & \\
& \ddots  & & & {\fontsize{50}{60}\selectfont \textsf{0}} & \\
& & 0 & & & \\
\hline
& & & -\delta_{1} & & \\
& {\fontsize{50}{60}\selectfont \textsf{0}} & & & \ddots & \\
& & & & & -\delta_{n+1} \\
\hline
0 & \dots & 0 & 0 & \dots & 1
\end{array}\right] \text{ with } 0 \leq \delta_{i} \leq 1, \thickspace \forall \thinspace 1\leq i\leq n+1, \label{gap-error} \\
 {\rm and\ }v_{k} = \dfrac{\sum_{i=0}^{n}B_{i}\tilde{p}_{k}^{i}}{\sum_{i=0}^{n}\WE(\cC)_{i}\tilde{p}_{k}^{i}}. \label{vk-WE-ratio}
\end{gather}
 
Denote $\vec\omega_0$ the solution of this linear system in the case where $\Delta = 0$, i.e., $\vec{\omega}_{0} = \sfM^{-1}\cdot \vec{b}$. This solution is not the weight enumerator, because we omitted the $\Delta\cJ$ term. However, since the weight enumerator coefficients are integers, we expect that for $\Delta$ sufficiently small, rounding up each component of  $\vec \omega_0$ to the nearest integer will give the correct answer $\vec \omega$. For this approach to succeed, it suffices to ensure that $|\vec{\omega}_{i} - \vec{\omega}_{0,i}| \leq \slfrac{1}{2}$. Using Eq.~({}\ref {eq-M-inv-gap}{}) for $\vec{\omega}$ and a similar one for $\vec{\omega}_{0}$ (with $\Delta = 0$), we can rewrite this condition in terms of the constraint matrix in Eq.~({}\ref {constraint-matrix}{}) as:
\begin{gather}
|\vec{\omega}_{i} - \vec{\omega}_{0,i}| = \left|\left(\left[\thinspace\left[\mathbb{I} + \Delta\thickspace \cJ\right]\cdot\sfM \thinspace\right]^{-1}_{i,n} - [\thinspace \sfM^{-1}\thinspace]_{i,n}\right)2^{n}\right| \leq \frac 12\label{gap-effect-WE}
\end{gather}
where we have exploited the special structure of $\vec{b}$ which ensures that only the elements in the last column of the inverse matrix have a non-trivial contributions to the components of $\vec{\omega}$.

Using properties of matrix multiplication and the $\infty$-norms of matrices in Eq.~({}\ref {gap-effect-WE}{}), we obtain
\begin{gather}
\left|\thinspace \left[\thinspace\left[\mathbb{I} + \Delta\thickspace \cJ\right]\cdot\sfM \thinspace\right]^{-1} - \sfM^{-1} \thinspace\right|_{i,n} \thickspace \leq \left(\left|\left|\thinspace \left[\mathbb{I} + \Delta\thickspace \cJ\right]^{-1} \thinspace\right|\right|_{\infty} - 1\right) \thickspace \left|\left[\sfM^{-1}\right]_{i,n}\right| .\label{eq-inv-norm}
\end{gather}
The explicit form of the matrix in Eq.~({}\ref {gap-error}{}) can be used to verify that $\left|\left|\thinspace (\thinspace\mathbb{I} + \Delta\thinspace \cJ\thinspace)^{-1} \thinspace\right|\right|_{\infty} \leq (1 - \Delta)^{-1}$ and also $[\sfM^{-1}]_{i,n}\thinspace 2^{n} = \WE_{i}(\cC) \leq n^{i}, \thickspace\forall \thinspace 1\leq i\leq n+1$. Hence we can now state a sufficient condition for Eq.~({}\ref {gap-effect-WE}{}) to hold:
\begin{gather}
\left[\dfrac{1}{1 - \Delta} - 1\right] \thinspace n^{i} \thickspace \leq \thickspace \dfrac{1}{2}, \quad \forall \thinspace 1\leq i\leq n+1\label{gap-error-WEi}
\end{gather}
The lemma follows by imposing this condition for all $i\leq \lambda$, yielding
\begin{gather}
\left[\dfrac{1}{1 - \Delta} - 1\right] \thinspace n^{\lambda} \thickspace \leq \thickspace \dfrac{1}{2} \thickspace \Rightarrow \thickspace \Delta \leq \dfrac{2}{2 + n^{\lambda}}. \label{gap-qpoly}
\end{gather}
\end{proof}

Though it is clear that we must have $2n+1$ linearly independent constraints on $\{\WE_{i}(\cC),B_{i}\}_{i=0}^{n}$ along with a promise gap of $\Delta \leq 2[2 + n^{\lambda}]^{-1}$, it is not immediately clear how many distinct crossing points need to be located to obtain these constraints, i.e. some of the crossing points associated to distinct pairs $(v_l,p_l)$ may result in  linearly dependent constraints Eq.~({}\ref {CPv}{}). The following lemma ensures that we can always efficiently find linearly independent constraints.

\begin{lemma} \label{indep-cons}
The identification of at most $8n^2$ distinct crossing points $(v_l,p_l)$ of Eq.~({}\ref {CPv}{}) will produce $2n+1$ linearly independent constraints.
\end{lemma}
\begin{proof}
The $(2n+2)\times (2n+2)$ matrix in Eq.~({}\ref {constraint-matrix}{}) can be represented as $\sfM = [\sfV\thinspace |\thinspace \sfD\cdot\sfV]$ where $\sfV_{i,j} = \tilde{p}_{i}^{j}$ is a $(2n+2) \times (n+1)$ \emph{Vandermonde matrix} and $\sfD_{i,j} = -v_{i}\delta_{i,j}$ is a diagonal square matrix. We shall prescribe a polynomial-time method to construct a full-rank matrix of this form.

Let $\mathsf{T}_{l}$ represent the sub-matrix of $\sfM$ formed by taking the first $l$ column entries of the first $l$ rows. Hence $\mathsf{T}_{l}$ is a square matrix. Moreover, notice that a sufficient condition for the the matrix in Eq.~({}\ref {constraint-matrix}{}) to describe $l$ independent constraints is that $\mathsf{T}_{l}$ must have a non-vanishing determinant. Hence, it suffices to show that each new constraint can be generated such that the corresponding square matrix $\mathsf{T}_{l}$ satisfy $\det(\mathsf{T}_{l}) \neq 0$.

Consider the simple case when $l \leq n+1$. In this case, $\mathsf{T}_{l}$ is always given by a \emph{Vandermonde matrix}, which is full rank if and only if all $p_{i}$ are distinct \cite{mat-book-uw-08}. Hence, this implies that $n+1$ distinct values of $v_{i}$ suffice to generate $n+1$ constraints that are necessarily independent, since different $v$'s necessarily produce different $p$'s (see Lemma.~({}\ref {crossing-acc}{})).

For $l > n+1$, it is not immediately clear if the procedure prescribed in Sec.~({}\ref {sec-stab-cons}{}) will  provide linearly independent constraints even if we choose all $v_{i}$ to be distinct. However, a sufficient condition for independency is again the non-vanishing of $\det(\mathsf{T}_{l})$. Assume that the choices of $\{v_{1}, \dots, v_{l-1}\}$ are such that $\det(\mathsf{T}_{l-1}) \neq 0$. The Laplace-expansion of the determinant suggests that $\det(\mathsf{T}_{l}) = 0$ when $v_{l}$ is given by:
\begin{gather}
v_{l} = \dfrac{\sum_{i=0}^{n}E_{i}\tilde{p}_{l}^{i}}{\tilde{p}_{n+t}^{t-1}\det(T_{l-1}) + \sum_{j=0}^{t-2}\tilde{p}_{n+j}^{j}C_{j}} \label{indep-det0}
\end{gather}
where $l > n + 1$ and $E_{i}, C_{i}$ are constants (independent of $p_{n+i}$; they are minors of $p^{i}$ and $vp^{i}$ respectively). 

The above equation classifies the particular values of $v_{l}$ that fail to generate $l$ linearly independent constraints on $\{\WE_{i}(\cC), B_{i}\}_{i=0}^{n}$. Note that $(v_{l},p_{l})$ are related by Eq.~({}\ref {vk-WE-ratio}{}). Hence, equating this with Eq.~({}\ref {indep-det0}{}) results in a rational polynomial in $p_{l}$, which has at most $2n$ roots, unless the expressions in Eqs.\nobreakspace \textup {(\ref {vk-WE-ratio})} and\nobreakspace  \textup {(\ref {indep-det0})} are identical. Assume for now that these expressions are not identical. Then we can choose to sample $2n+1$ different values of $v_{l}$ according to Eq.~({}\ref {vk-WE-ratio}{}), and for at least one of them, the resulting $p_{l}$ will violate Eq.~({}\ref {indep-det0}{}), so the pair $(v_{l},p_{l})$ will produces a linearly independent constraint. On the other hand, if the two functions in Eqs.\nobreakspace \textup {(\ref {vk-WE-ratio})} and\nobreakspace  \textup {(\ref {indep-det0})} are identical, then we have that $\forall p$:
\begin{gather}
\dfrac{\sum_{i=0}^{n}\WE_{i}(\cC)\tilde{p}^{i}}{\sum_{i=0}^{n}B_{i}\tilde{p}^{i}} = \dfrac{\sum_{i=0}^{n}E_{i}\tilde{p}^{i}}{\tilde{p}^{k-1}\det(T_{n+k-1}) + \sum_{j=0}^{k-2}\tilde{p}^{j}C_{j}}. \label{indep-fix-WE}
\end{gather}
The above relation immediately fixes the coefficients $\{\WE_{i}(\cC), B_{i}\}_{i=0}^{n}$ by just comparing the powers of $\tilde{p}$. 

Since we must sample at most $2n+1$ crossing points to obtain a linearly independent constraint on $\{\WE_{i}(\cC), B_{i}\}_{i=0}^{n}$, we require at most $(2n+1)(2n+1)$ or (for expressional convenience) $8n^{2}$ distinct crossing points to construct the a full rank matrix of the form in Eq.~({}\ref {constraint-matrix}{}).
\end{proof}

The last ingredient we need is a bound on the distance between crossing points. Remember that we are only able to locate the value of a crossing probability $p_k$ to exponential accuracy. Thus, it is necessary that changing the tunable parameter $v$ has a significant effect on the value of the crossing point in order to generate linearly independent constraints (with significantly different values of $p_k$). Combining the restriction on the values of the tunable parameter in Eq.~({}\ref {eq-max-crossing}{}) with Lemma.~({}\ref {indep-cons}{}) immediately tells that the smallest change in the tunable parameter will be at least $(1-\Delta)n^{-n}8^{-1}n^{-2}$. This naturally implies a minimum separation between two crossing points, as the lemma below addresses:

\begin{lemma} \label{crossing-acc}
Let $\{v_k,p_k\}_{k=1}^{2n+1}$ denote crossing points of a stabilizer code constructed as above with $|v_k-v_l|\geq [8n^{2}(1-\Delta)n^{n}]^{-1}$. Then $|p_{i} - p_{j}| \geq 4^{-n\log_{2}n}$.
\end{lemma}
\begin{proof}
We need to relate the difference in $v$ to a difference in $p$. For this, let us recall that if $\gamma$ is a small change in $v$ and $\delta$ a small change in $p$, then $\gamma = \delta \thinspace (\slfrac{dv}{dp})$, or equivalently: $\gamma (\slfrac{dv}{dp})^{-1} = \delta$. Computing the derivative using the expression for $v$ in Eq.~({}\ref {vk-WE-ratio}{}), we have:
{\small
\begin{flalign}
\delta &= \gamma\cdot \left(\dfrac{\sum_{i,j=0}^{n}\WE_{i}(\cC)\WE_{j}(\cC)\left(\dfrac{p}{1-p}\right)^{i+j}}{\sum_{i,j=0}^{n}B_{i}\WE_{j}(\cC)\left[\left(\dfrac{p}{1-p}\right)^{j}\dfrac{d}{dp}\left(\dfrac{p}{1-p}\right)^{i} - \left(\dfrac{p}{1-p}\right)^{i}\dfrac{d}{dp}\left(\dfrac{p}{1-p}\right)^{j}\right]}\right) \nonumber \\
&\geq \gamma \cdot \left(\dfrac{1}{\dfrac{1}{(1-p)^{2}}\sum_{i,j=0}^{n}\left[iB_{i}\WE_{j}(\cC)\left(\dfrac{p}{1-p}\right)^{j}\left(\dfrac{p}{1-p}\right)^{i}\right]}\right) \nonumber \\
&\geq \gamma \cdot \dfrac{(1-p)^{2n+2}}{n} \thickspace \Rightarrow \thickspace \delta \geq \gamma \cdot \dfrac{4^{1-n}}{n}
\end{flalign}
}
Since $\gamma \geq (1-\Delta)n^{-n}8^{-1}n^{-2}$, we find: $\delta \geq (1-\Delta)n^{-n}8^{-1}n^{-3}4^{1-n} \geq 4^{-n\log_{2}n}$.
\end{proof}
Hence it suffices to estimate the location of each crossing point to within an accuracy of $4^{-n\log_{2}n}$, implying that we must run at most $2n\log_{2}n$ iterations of the above binary-search like procedure, mentioned at the end of Sec.~({}\ref {sec-crossing}{}), to locate the crossing point. Assuming that each query to a $\DQMLD$ oracle is answered in constant time, the above procedure takes $\cO(n\log_{2}n)$ time.

In appendix App.~({}\ref {app-concatenate}{}) we construct a repetition code with polynomially many qubits to encode a single qubit, such that the resulting noise on the encoding qubit obeys $p_Z = q$, $p_I = 1-q$, and $p_Z=p_Y =0$ with an exponential accuracy (the error-free model is a special case $q=0$). This implies that the tunable parameters $v_k$ used to find crossing points can be set to exponential accuracy. 

\vspace{0.3cm}

\noindent With these, we are in a position to prove our main result. 

\vspace{-0.2cm}

\begin{proofMainThm}
Given an input classical linear code $\cC$, we built a generating set for a stabilizer code, with two logicals, that satisfies Eq.~({}\ref {prob-I-WE}{}). Appending additional qubits to this code and choosing a specific channel Sec.~({}\ref {sec-stab-cons}{}), on those qubits, enabled the introduction of a tunable parameter in the logical class probabilities Eqs.\nobreakspace \textup {(\ref {stab-cons-I})} and\nobreakspace  \textup {(\ref {stab-cons-Z})}. Varying this tunable parameter as per Lemma.~({}\ref {gap-size-qpoly}{}), with a polynomial number of queries to a $\DQMLD$ oracle, we can estimate the location of a crossing point to within exponential accuracy, thereby providing a linear constraint on the weight enumerator coefficients. Repeating this procedure $2n+1$ times, we showed in Lemma.~({}\ref {lemma-gap-size}{}) that, as long as the promise gap is $2[2+n^{\lambda}]^{-1}$, the system of inequalities yield the same solution to the first $\lambda$ weight enumerator coefficients of $\cC$, up to an integer approximation, as the system of equalities considered without any promise gap. Since $\lambda = \textsf{polylog}(n)$ in Thm.~({}\ref {main-theorem}{}), it suffices to have a promise gap in $\DQMLD$ that is $1/\textsf{quasi-polynomial(n)}$. It followed from Lemma.~({}\ref {indep-cons}{}) that at most $8n^{2}$ crossing points are sufficient to generate the necessary linearly independent constraints on the weight enumerator coefficients and from Lemma.~({}\ref {crossing-acc}{}) that the accuracy on their location, achievable in polynomial time, is sufficient. Lastly, we showed in App.~({}\ref {app-concatenate}{}) that it suffices to direct all queries to a $\DQMLD$ oracle on a independent $X-Z$ channel. Thus we prove Thm.~({}\ref {main-theorem}{}).
\end{proofMainThm}

\section{Conclusion}
We will close with mentioning a few open problems which we were not able to address in this paper. In the course of this paper we have addressed the optimal decoding problem on an independent $X-Z$ channel. However, the same can be done for a depolarizing channel by introducing the notion of a generalized weight described in \cite{yueh-13}. Hence, Sec.~({}\ref {sec-stab-cons}{}) of the paper will undergo certain modifications when choosing to address a depolarizing channel.

The key problem turns out to be the classification of the \emph{parametrized} complexity of the decoding with the promise gap parameter Def.~({}\ref {def-DQMLD}{}), denoted by $\Delta$. We know the complexities for two extreme cases, namely $\DQMLD$ is $\sharpPComplete$ when $\Delta = \slfrac{1}{\textsf{quasi-polynomial}(n)}$, c.f. Thm.~({}\ref {main-theorem}{}) and in $\NP$ when $\Delta = 1 - 2^{-n-k}$, c.f. Lemma.~({}\ref {lemma-large-gap}{}). However, for a vast intermediate range of $\Delta$, complexity of $\DQMLD$ remains open. As described in Sec.~({}\ref {sec-deg-dec}{}), for a code encoding a single qubit, the promise gap $\Delta$ is related to the decoding failure probability $\epsilon$ as $\Delta = 1-\epsilon$. Thus, the two extreme cases considered above correspond respectively to optimal decoding in a very noisy regime (failure probability approaching unity) and optimal decoding with an exponentially small failure probability. Unfortunately, the case of practical interest falls somewhere in between. 

The complexity of any problem is only a highlight of the runtime of any algorithm on the worst case instance of the problem. Hence it is of practical interest to know the runtime of the algorithm for any \emph{typical} instance. This could refer for instance to a typical syndrome or a random code. It is well known however that random codes are non-degenerate \cite{got-phd-97,ashik00}, so the decoding is not expected to be affected by the degeneracy in errors, so our result is probably not relevant in this setting. However, for the practically relevant class of \emph{sparse codes}, the complexity of optimal decoding strategy remains an important open question.

Lastly, our analysis has focused on stabilizer codes over Pauli channels. This is particularly convenient due to the discrete nature of the resulting decoding problem. This setting could be generalized in two obvious ways. First, we could consider codes that are not stabilizer codes, defined from a set of commuting projectors. There exists a growing interest for those codes, particularly in the setting of topological quantum order \cite{Kit03a,HP10a,LP13a,BD12a}. In this setting, we could study for instance the decoding problem for systems that support non-Abelian anyons \cite{BBDFP13}.  Second, we could consider errors that are not described by Pauli operators. This problem is of practical importance because no real-world device undergoes a Pauli channel; for instance the physical process of relaxation is not described by a Pauli channel. 

\section{Acknowledgements}
We thank Daniel Gottesman for stimulating discussions and Guillaume Duclos-Cianci for careful reading of this manuscript. This work was partly funded by Canada's NSERC and by the Lockheed Martin Corporation.

\clearpage


\clearpage

\appendix
\section{Simulating Pauli channels with an independent $X-Z$ channel} \label{app-concatenate}

In our reduction, we have used two features of the a general memoryless Pauli channel, that are not intrinsic to an independent $X-Z$ channel. They involve the freedom of assigning unequal noise rates for $\mathbb{I}, X,Y$ and $Z$ type errors on some qubits of the code. However, when a qubit of the code is encoded into an auxiliary code, the the optimal decoder on this concatenated code will replace the physical noise rates for the qubit with the logical noise rates (logical class probabilities) of the auxiliary code \cite{poulin-06}. Therefore, the physical noise rates for that qubit can be controlled by controlling the logical class probabilities of the auxiliary code. The latter can be achieved by varying the syndrome on the auxiliary code.

In our case, the auxiliary code is the Shor code \cite{Sho95a,preskill-13} on a $n_{1} n_{2}$ qubits, with $n_{1}, n_{2} \sim \poly(n)$. The stabilizers of this code can be represented graphically on a $n_{1}\times n_{2}$ lattice with a qubits at each vertex, see Fig.~({}\ref {sim-shor-code}{}). Each horizontal link between a pair of vertices represents a $X-$type stabilizer on the qubits corresponding to the vertices. A pair of rows represents a $Z-$type stabilizer on the qubits corresponding to the vertices on those rows.
\begin{figure}[H]
\centering
\subfloat[Stabilizers and logicals of the Shor code on a $n_{1}\times n_{2}$ lattice.]{\label{sim-shor-code}\includegraphics[scale=0.26]{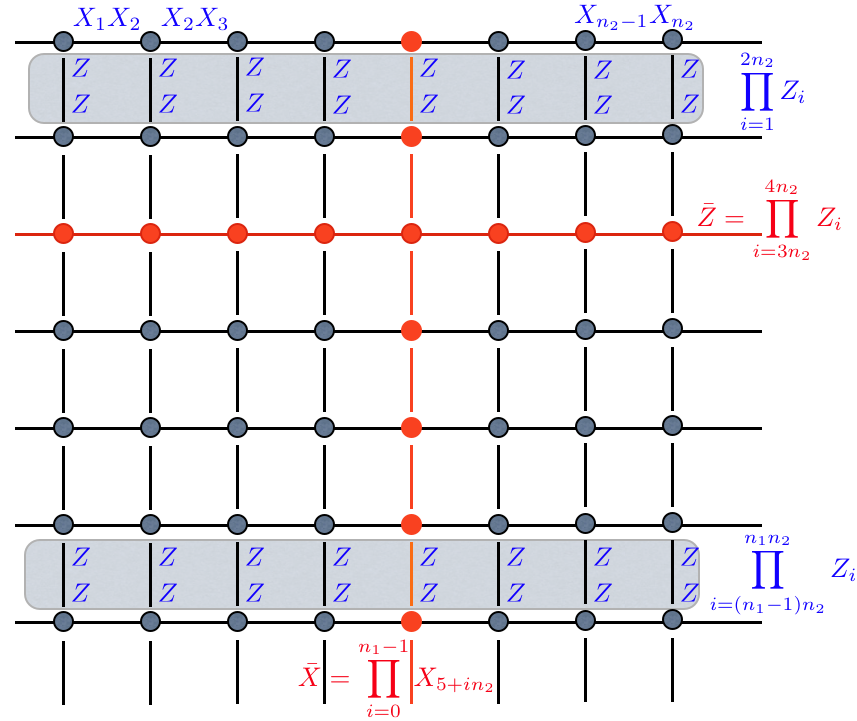}}
\hspace{0.05cm}
\hspace{0.05cm}
\subfloat[Syndrome corresponding to a $Z-$type error on qubit along a row.]{\label{sim-shor-errors}\includegraphics[scale=0.27]{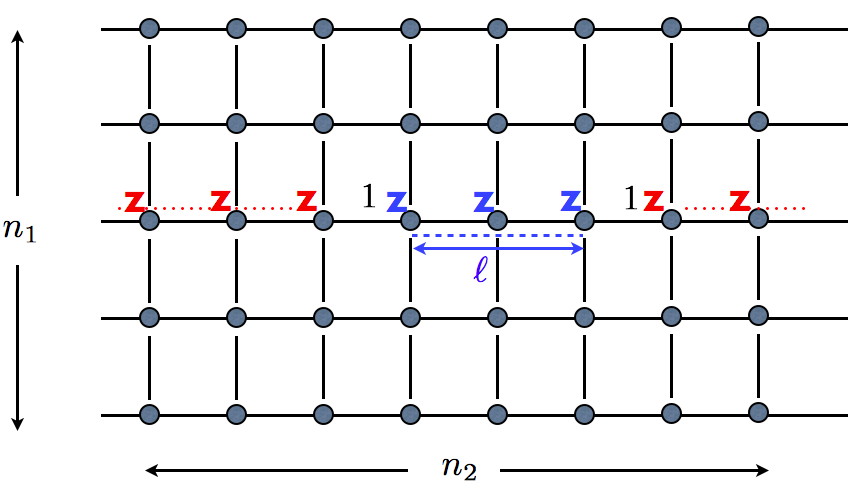}}
\caption{Shor code on a $n_{1}\times n_{2}$ lattice.}
\label{fig-toric-code-errors}
\end{figure}

Any $Z-$type error chain will turn on a syndrome, represented by a pair of points on two ends of the chain, c.f. Fig.~({}\ref {sim-shor-errors}{}). Let us denote the smallest number of links between the points as $\ell$. We will restrict to the case where there is a single continuous error chain, which is confined to a row of the lattice. Hence $\ell$ completely defines the syndrome. For the syndrome in (Fig. \ref{sim-shor-errors}), one can immediately write the probabilities for the four logical classes $\overline{\mathbb{I}}$, $\overline{X}$, $\overline{Z}$ and $\overline{Y}$ as:
\begin{gather}
\Prob(\overline{\mathbb{I}}|\vec{s}) = \dfrac{\left[\dfrac{p}{2}\left(1 - \dfrac{p}{2}\right)\right]^{\ell}\left[\left(1-\dfrac{p}{2}\right)\left(1-\dfrac{p}{2}\right)\right]^{n_{2}-\ell}}{\tilde{P}(s)} \label{shor-prob-I} \\
\Prob(\overline{Z}|\vec{s}) = \dfrac{\left[\dfrac{p}{2}\left(1 - \dfrac{p}{2}\right)\right]^{n_{2}-\ell}\left[\left(1-\dfrac{p}{2}\right)\left(1-\dfrac{p}{2}\right)\right]^{\ell}}{\tilde{P}(s)} \label{shor-prob-Z} \\
\Prob(\overline{X}|\vec{s}) = \dfrac{\left[\dfrac{p}{2}\left(1 - \dfrac{p}{2}\right)\right]^{\ell}\left[\left(1-\dfrac{p}{2}\right)\left(1-\dfrac{p}{2}\right)\right]^{n_{2}-\ell}\sum_{i=1}^{n_{2}}\binom{n_{2}}{i}\left(\dfrac{p}{2}\right)^{n_{1} i}\left(1-\dfrac{p}{2}\right)^{2n_{1}n_{2} - n_{1}i}}{\tilde{P}(s)} \label{shor-prob-X} \\
\Prob(\overline{Y}|\vec{s}) = \dfrac{\left[\dfrac{p}{2}\left(1 - \dfrac{p}{2}\right)\right]^{n_{2}-\ell}\left[\left(1-\dfrac{p}{2}\right)\left(1-\dfrac{p}{2}\right)\right]^{\ell}\sum_{i=1}^{n_{2}}\binom{n_{2}}{i}\left(\dfrac{p}{2}\right)^{n_{1} i}\left(1-\dfrac{p}{2}\right)^{2n_{1}n_{2} - n_{1}i}}{\tilde{P}(s)}  \label{shor-prob-Y}\\
\text{where: } \tilde{P}(\vec{s}) = \left[\left(\dfrac{p}{2}\right)^{\ell}\left(1-\dfrac{p}{2}\right)^{2n_{2}-\ell} + \left(\dfrac{p}{2}\right)^{n_{2}-\ell}\left(1-\dfrac{p}{2}\right)^{2\ell}\right]\sum_{i=0}^{n_{2}}\binom{n_{2}}{i}\left(\dfrac{p}{2}\right)^{n_{1} i}\left(1-\dfrac{p}{2}\right)^{2n_{1}n_{2} - n_{1}i} \label{shor-prob-synd}
\end{gather}

Given a constant $v$, we wish to simulate a channel, on the encoded qubit, with $p_{I} = 1-q, \thinspace p_{Z} = q$ such that $\slfrac{(1-q)}{q} = v$. This implies a ratio between $\Prob(\overline{\mathbb{I}}|\vec{s})$ and $\Prob(\overline{Z}|\vec{s})$ equal to $v$, and $\ell$ given by
\begin{gather}
v := \dfrac{\Prob(\thinspace\overline{\mathbb{I}}\thinspace|\thinspace s\thinspace)}{\Prob(\thinspace\overline{Z}\thinspace|\thinspace s\thinspace)} = \left[\dfrac{p}{2(1-p)}\right]^{2\ell}\thinspace \left[\dfrac{2(1-p)}{p}\right]^{n_{2}} \thickspace \Rightarrow \thickspace \ell = \dfrac{1}{2}\left(n_{2} + \dfrac{\left|\log_{2}v\right|}{1 + \log_{2}\frac{1-p}{p}}\right) \label{shor-synd}
\end{gather}
where in the last simplification, we have assumed $v\leq 1$. As a consequence of the upper bound on $v$ in Eq.~({}\ref {eq-max-crossing}{}) it suffices to choose $n_{2} = 2n$. We will fix the number of columns in the auxiliary code Fig.~({}\ref {sim-shor-code}{}) to $2n$.

In addition to the ratio between $p_{\mathbb{I}}$ and $p_{Z}$, the channel on the encoded qubit also requires that $p_{X} = p_{Y} = 0$. To achieve this, we must choose a value of $n_{1}$ such that the expressions in Eqs.\nobreakspace \textup {(\ref {shor-prob-X})} and\nobreakspace  \textup {(\ref {shor-prob-Y})} are vanishingly small. Note that the ratio of the combinatorial sums in the the expressions for $\Prob(\overline{Y}|\vec{s})$ and $\tilde{P}(s)$ can be bounded from above as:
\begin{flalign}
\dfrac{\sum_{i=1}^{n_{2}}\binom{n_{2}}{i}\left(\dfrac{p}{2}\right)^{n_{1} i}\left(1-\dfrac{p}{2}\right)^{2n_{1}n_{2} - n_{1}i}}{\sum_{i=0}^{n_{2}}\binom{n_{2}}{i}\left(\dfrac{p}{2}\right)^{n_{1} i}\left(1-\dfrac{p}{2}\right)^{2n_{1}n_{2} - n_{1}i}} &\leq \left(\frac{p}{2-p}\right)^{n_{1}} + \dfrac{n_{2}^{2}\left(\frac{p}{2-p}\right)^{2n_{1}}}{1 - n_{2}\left(\frac{p}{2-p}\right)^{n_{1}}} \\
&\leq (2n-1)^{-n_{1}} + \dfrac{n^{2}_{2}(2n-1)^{-2n_{1}}}{1 - n_{2}(2n-1)^{-n_{1}}}
\label{shor-ZY-sums}
\end{flalign}
Hence it suffices to take $n_{1} = 2n$ for the above ratio to be bounded above by a vanishingly small number. Moreover, as the ratio is an upper bound for $\Prob(\overline{X}|\vec{s})$ and $\Prob(\overline{Y}|\vec{s})$, we will fix the number of rows in the auxiliary code Fig.~({}\ref {sim-shor-code}{}) to $2n$.

Summarizing, the qubit is encoded into a Shor code on a $2n\times 2n$ lattice Fig.~({}\ref {fig-toric-code-errors}{}) and the syndrome is chosen as indicated in Fig.~({}\ref {sim-shor-errors}{}) with $\ell$ specified by Eq.~({}\ref {shor-synd}{}). As a result we have a channel on the qubit with $p_{\mathbb{I}} = 1-q, \thinspace p_{Z} = q, \thinspace p_{X} = p_{Y} = 0$, to within exponential error bars. Note that an error free channel is a special case of the above channel, with $q = 0$. Hence, we repeat the same encoding but choose the syndrome on the Shor code to be trivial.

\section{Proof of Lemma.~({}\ref {lemma-large-gap}{})} \label{app-large-gap}
In this appendix, we present a proof of Lemma.~({}\ref {lemma-large-gap}{}) which states that with a promise gap $\Delta \geq 1-2^{-n-k}$ and on an independent $X-Z$ channel, the outputs of $\QMLD$ and $\DQMLD$ are equivalent.

It suffices to demonstrate that, for any syndrome $\vec{s}$, the logical class containing the minimum weight error $L_{\min}$ must satisfy Eq.~({}\ref {eq-DQMLD-gap}{}) with $L^{\star} = L_{\min}$. This will be true whenever $1 - \Prob(L_{\min}|\vec{s}) \leq \Delta = 1 - 2^{-n-k}$, or in other words 
\begin{gather}
\dfrac{\Prob(L_{\textsf{min}} , \vec{s})}{\Prob(\vec{s})} \geq 2^{-n-k}. \label{app-gap-cond}
\end{gather}
Suppose that $a(\vec{s})$ is the minimum weight of an error, consistent with the syndrome $\vec{s}$. Let us derive lower and upper bounds separately for the numerator and the denominator of Eq.~({}\ref {app-gap-cond}{}), respectively. Every element in the logical class of $L_{\textsf{min}}$ is constructed by taking the product of the minimum weight error with every stabilizer in $\cS$. To find the lowest possible value of the probability of the product, we suppose that each product results in an element whose weight is the sum of $a(\vec{s})$ and that of the corresponding stabilizer. Hence we have:
\begin{gather}
\Prob(L_{\textsf{min}},\vec{s}) \geq \left[1 - \dfrac{p}{2}\right]^{2n}\sum_{i = 0}^{n}B_{i}\tilde{p}^{i + a(\vec{s})} \label{app-gap-Lmin}
\end{gather}
where $\tilde{p} = \slfrac{p}{(2-p)}$ and $B_{i} = |\{g\in\cS : \wtt(g) = i\}|$ is a weight enumerator coefficient of $\cS$.

Similarly, all the errors consistent with the syndrome $\vec{s}$ can be constructed by taking the error of weight $a(\vec{s})$ and taking its product with every element of the normalizer $\cN(\cS)$. For an upper bound to the probability of this product, let us suppose that each product has a weight equal to the difference between $a(\vec{s})$ and the weight of the corresponding stabilizer. As the weight of the product cannot be lower than $a(\vec{s})$ itself, this gives:
\begin{gather}
\Prob(\vec{s}) \leq \left[1 - \dfrac{p}{2}\right]^{2n}\left[\sum_{i=0}^{2a(\vec{s})}B^{\perp}_{i}\tilde{p}^{a(\vec{s})} + \sum_{i=2a(s)+1}^{n}B^{\perp}_{i}\tilde{p}^{i - a(\vec{s})}\right] \label{app-gap-synd}
\end{gather}
where $B^{\perp}_{i} = |\{g\in\cN(\cS) : \wtt(g) = i\}|$ is a weight enumerator coefficient of $\cN(\cS)$ \cite{got-phd-97}.

The ratio of expressions in Eqs.\nobreakspace \textup {(\ref {app-gap-Lmin})} and\nobreakspace  \textup {(\ref {app-gap-synd})}, bounds the quantity in Eq.~({}\ref {app-gap-cond}{}) from above, by:
\begin{gather}
\dfrac{\Prob(L_{\textsf{min}} , \vec{s})}{\Prob(\vec{s})} \geq \dfrac{1}{\sum_{i=0}^{2a(\vec{s})}B^{\perp}_{i} + \tilde{p}\sum_{i=2a(\vec{s})+1}^{n}B^{\perp}_{i}} \geq \dfrac{1}{\sum_{i=0}^{n}B^{\perp}_{i}} = 2^{-n-k}
\end{gather}
thereby proving the lemma. \hfill $\square$
\end{document}